\documentclass[12pt]{article}
\usepackage[pdftex]{graphicx}
\usepackage[T2A]{fontenc}
\usepackage{amsthm,amsmath,amssymb}

\allowdisplaybreaks

\usepackage[usenames,dvipsnames]{color}
\usepackage[pdftex,colorlinks=true,unicode,bookmarksnumbered=true]{hyperref}

\usepackage[dvipsnames]{xcolor}

\hypersetup{
  colorlinks,
  citecolor=blue!80!black,
  linkcolor=red!80!black,
  urlcolor=blue!80!black}

\usepackage[onehalfspacing]{setspace}
\usepackage{textcomp}
\usepackage{indentfirst}
\usepackage{natbib}
\usepackage{listings}
\usepackage{pdflscape}
\usepackage{enumerate}
\usepackage{hyperref}
\usepackage{verbatim}
\usepackage{bbm}
\usepackage{tikz}
\usepackage{pgfplots}
\usetikzlibrary{arrows,shapes,trees}
\usepackage{float}
 \usepackage{multirow}
 \usepackage{authblk}
 
 \usepackage[toc,page]{appendix}

\usepackage{array,longtable}

\usepackage{mathtools}

\DeclareMathOperator*{\Card}{Card}

\numberwithin{equation}{section}
\newtheorem{theorem}{Theorem}

\newtheorem{assumption}{Assumption}

\newtheorem{lemma}{Lemma}

\theoremstyle{definition}
\newtheorem{remark}{Remark}

\renewenvironment{abstract}
 {\small
  \begin{center}
  \bfseries \abstractname\vspace{-.5em}\vspace{0pt}
  \end{center}
  \list{}{%
    \setlength{\leftmargin}{5mm}
    \setlength{\rightmargin}{\leftmargin}%
    \listparindent 1.5em
    \itemindent    \listparindent
  }%
  \item\relax}
 {\endlist}

\newcommand{\lsim}{\mathrel{\hbox{\rlap{\lower.75ex \hbox{$\sim$}} \kern-.3em \raise.4ex \hbox{$<$}}}}

\usepackage{natbib}
\usepackage{geometry}
\begin{document}

\title{A Consistent LM Type Specification Test \\ for Semiparametric Panel Data Models\thanks{I thank seminar participants at several universities and conferences for helpful comments. All remaining errors are mine.}}
\author{Ivan Korolev\thanks{Department of Economics, Binghamton University. E-mail: \mbox{\href{mailto:ikorolev@binghamton.edu}{ikorolev@binghamton.edu}}. Website: \url{https://sites.google.com/view/ivan-korolev/home}}}
\date{September 12, 2019}

\maketitle

\begin{abstract}
This paper develops a consistent series-based specification test for semiparametric panel data models with fixed effects. The test statistic resembles the Lagrange Multiplier (LM) test statistic in parametric models and is based on a quadratic form in the restricted model residuals. The use of series methods facilitates both estimation of the null model and computation of the test statistic. The asymptotic distribution of the test statistic is standard normal, so that appropriate critical values can easily be computed. The projection property of series estimators allows me to develop a degrees of freedom correction. This correction makes it possible to account for the estimation variance and obtain refined asymptotic results. It also substantially improves the finite sample performance of the test.
\end{abstract}


\newpage

\doublespacing

\section{Introduction}\label{introduction}

Panel data allows researchers to better account for individual heterogeneity and estimate richer models than cross-sectional data. Traditionally, the literature on panel data models focused on fully parametric models. Popular textbooks written by \citet{arellano_2003}, \citet{hsiao_2003}, and \citet{baltagi_2013} give excellent overviews of such models. However, parametric models may be misspecified, and more flexible panel data models may be needed.

Semiparametric models, such as partially linear or varying coefficient models, serve as an attractive alternative to fully parametric models. While being more flexible than parametric models, they are more tractable than fully nonparametric models and alleviate the curse of dimensionality. \citet{ai_li_2008}, \citet{su_ullah_2011}, \citet{rodriguez-poo_soberon_2017}, and \citet{parmeter_racine_2018} provide excellent surveys of recent developments in semiparametric panel data models.

While there is a growing literature on estimation of semiparametric panel data models, the literature on specification testing in this setting remains scarce. There are three possible reasons for this. First, in the presence of fixed effects, one needs to transform data to eliminate them before the model can be estimated. Estimating the transformed model in itself can be challenging when kernel methods are used. 

Second, the asymptotic theory for consistent specification tests in semiparametric panel data models can be challenging. For instance, \citet{henderson_et_al_2008} propose kernel-based specification tests both for parametric and semiparametric fixed effects panel data models and suggest using the bootstrap to obtain critical values, but they do not derive asymptotic properties of their tests. In turn, \citet{lin_et_al_2014} develop an asymptotic theory for kernel-based specification tests for panel data models with fixed effects, but they only consider parametric models.

Third, it has long been known in the literature on consistent specification tests that asymptotic approximations often do not work well in finite samples even with cross-sectional data (see, e.g., \citet{li_wang_1998}). The bootstrap is typically used to improve the finite sample performance of consistent specification tests, but it may be computationally costly.

In this paper I rely on the results on series estimation of fixed effects panel data models from \citet{baltagi_li_2002} and \citet{an_et_al_2016} and develop a consistent Lagrange Multiplier (LM) type specification test for semiparametric panel data models. My test overcomes all three challenges described above.

First, the use of series methods leads to a model that is linear in parameters. As a result, transforming the data, e.g. applying the within transformation or taking first differences, to eliminate fixed effects is straightforward. Thus, the test is simple to implement.

Second, as in cross-sectional models in \citet{korolev_2019}, the projection property of series estimators allows me to develop a degrees of freedom correction. Intuitively, when series methods are used, the restricted residuals are orthogonal to the series terms included in the restricted model. This means that even under the alternative, only a subset of moment conditions, rather than all of them, can be violated, which in turn affects the normalization of the test statistic. This degrees of freedom correction has two important consequences. 

From the theoretical point of view, it leads to a tractable asymptotic theory for the test and allows me to obtain refined asymptotic results. In my asymptotic analysis, I decompose the test statistic into the leading term and the remainder. By relying on the projection nature of series estimators, I can directly account for the estimation variance, so that only bias enters the remainder term. Because of this, I only need to control the rate at which bias goes to zero to bound the remainder term, while variance can remain large. As a result, I can derive the asymptotic distribution of the test statistic under fairly weak rate conditions.

From the practical point of view, the degrees of freedom correction substantially improves the finite sample performance of the test. While I propose a wild bootstrap procedure and establish its asymptotic validity, I show using simulations that the asymptotic version of the proposed test with the degrees of freedom correction performs almost as well as its wild bootstrap version. Hence, the degrees of freedom correction serves as a computationally attractive analytical way to obtain a test with good small sample behavior.

The remainder of the paper is organized as follows. Section~\ref{model_test} introduces the model and describes how to construct the series-based specification test for semiparametric fixed effects models. Section~\ref{asymptotics} develops the asymptotic theory for the proposed test. Section~\ref{simulations} studies the behavior of the proposed test in simulations. Section~\ref{empirical_example} applies my test to the data from \citet{cornwell_rupert_1988} and \citet{baltagi_khanti-akom_1990}. Section~\ref{conclusion} concludes.

\ref{appendix_tables_figures} collects all tables and figures. \ref{appendix_proofs} contains proofs of my results.

\section{The Model and Proposed Test}\label{model_test}

I consider a general nonparametric panel data model with fixed effects:
\begin{align}\label{np_model}
Y_{it} = g(X_{it}) + u_{it} = g(X_{it}) + \mu_i + \varepsilon_{it}, \quad E[\varepsilon_{it} | X_{i}, \mu_i ] = 0,
\end{align}
where $X_i = (X_{i1},...,X_{iT})'$, $t = 1, ..., T$, and $i = 1, ..., n$. $\mu_i$ denotes the fixed effect, which captures unobserved heterogeneity and may be correlated with the regressors $X_i$. In my asymptotic analysis, I will assume that $T$ is fixed while $n$ grows to infinity.

The goal of this paper is to test that the true model is semiparametric, i.e. that
\begin{align}\label{generic_h0}
H_0^{SP}: P_{X} \left(g(X_{it}) = f(X_{it}, \theta_0, h_0) \right) = 1 \text{ for some } \theta_0 \in \Theta, h_0 \in \mathcal{H},
\end{align}
where $f: \mathcal{X} \times \Theta \times \mathcal{H} \to \mathbb{R}$ is a known function, $\theta \in \Theta \subset \mathbb{R}^{d}$ is a finite-dimensional parameter, and $h \in \mathcal{H} = \mathcal{H}_1 \times ... \times \mathcal{H}_q$ is a vector of unknown functions. For instance, if the semiparametric model is partially linear, then $f(X_{it},\theta,h) = X_{1it}' \theta + h(X_{2it})$, where $X_{it} = (X_{1it}',X_{2it}')'$. Many other semiparametric models can also be written in this form.

The global alternative is
\begin{align}\label{generic_h1}
H_1: P_{X} \left(g(X_{it}) \neq f(X_{it}, \theta, h) \right) > 0 \text{ for all } \theta \in \Theta, h \in \mathcal{H}
\end{align}

\subsection{Series Estimators}

As in \citet{korolev_2019}, I use series methods to replace unknown functions with their finite series expansions. Namely, for any variable $z$, let $Q^{a_n}(z) = (q_1(z), ..., q_{a_n}(z))'$ be an $a_n$-dimensional vector of approximating functions of $z$, where the number of series terms $a_n$ is allowed to grow with the sample size $n$. Then an unknown function $g(z)$ can be approximated as $g(z) \approx \sum_{j=1}^{a_n}{q_j(z) \gamma_j} = Q^{a_n}(z)' \gamma$. I replace all unknown functions in $f(X_{it}, \theta, h)$ with their finite series expansions and write the semiparametric model in a series form as
\begin{align}\label{sp_model_series}
Y_{it} = W_{it}' \beta_1 + R_{it} + \mu_i + \varepsilon_{it},
\end{align}
where $W_{it} := W^{m_n}(X_{it}) := (W_{1}(X_{it}), ..., W_{m_n}(X_{it}))'$ are appropriate regressors or basis functions, such as power series or splines, $m_n$ is the number of parameters in the semiparametric null model, $R_{it} = f(X_{it}, \theta, h) - W_{it}' \beta_1$ is the approximation error.

\subsection{Test Statistic}

To construct a specification test, I include additional series terms, $Z_{it} := Z^{r_n}(X_{it}) := (Z_{1}(X_{it}), ..., Z_{r_n}(X_{it}))'$, that capture possible deviations from the null hypothesis:
\begin{align}\label{np_model_series}
Y_{it} = W_{it}' \beta_1 + Z_{it}' \beta_2 + R_{it} + \mu_i + \varepsilon_{it} = P_{it}' \beta + R_{it} + \mu_i + \varepsilon_{it},
\end{align}
where $P_{it} := P^{k_n}(X_{it}) := (W_{it}', Z_{it}')'$, $k_n = m_n + r_n$ is the total number of parameters, and $\beta = (\beta_1', \beta_2')'$. For instance, in the partially linear model example above, the additional series terms can include nonlinear terms in $X_{1it}$ and interactions between $X_{1it}$ and $X_{2it}$.

Due to the presence of fixed effects $\mu_i$ that may be correlated with $X_{it}$, it is problematic to estimate or test this model directly. Instead, I use first differencing or the within transformation to get rid of fixed effects. The model becomes
\[
\hat{Y}_{it} = \hat{W}_{it}' \beta_1 + \hat{Z}_{it}' \beta_2 + \hat{R}_{it} + \hat{\varepsilon}_{it} = \hat{P}_{it}' \beta + \hat{R}_{it} + \hat{\varepsilon}_{it},,
\]
where in the former case for any variable $A_{it}$, $\hat{A}_{it} = A_{it} - A_{i,t-1}$, and in the latter case $\hat{A}_{it} = A_{it} - \frac{1}{T}\sum_{s=1}^{T}{A_{is}}$. The specification test reduces to testing the hypothesis $\beta_2 = 0$.

For any variable $A_{it}$, let $A_i = (A_{i1},...,A_{iT})'$ and $A = (A_1', ..., A_n')'$. The restricted estimate of $\beta_1$ is obtained from the regression of $\hat{Y}$ on $\hat{W}$ and is given by
\[
\tilde{\beta}_1 = (\hat{W}' \hat{W})^{-1} \hat{W}' \hat{Y},
\]
and the restricted residuals are
\[
\tilde{e} = \hat{Y} - \hat{W} (\hat{W}' \hat{W})^{-1} \hat{W}' \hat{Y} = M_{\hat{W}} \hat{Y},
\]
where $M_{\hat{W}} = I - \hat{W} (\hat{W}' \hat{W})^{-1} \hat{W}'$. If the null is true, it can be shown that
\[
\tilde{e} = M_{\hat{W}} \hat{\varepsilon} + M_{\hat{W}} \hat{R}
\]

The test will be based on the moment condition $E[\hat{P}_i' \hat{\varepsilon}_i] = 0$. The sample analog of this moment condition is $\sum_{i=1}^{n}{\hat{P}_i' \tilde{e}_i}/n$. Note that $\sum_{i=1}^{n}{\hat{W}_i' \tilde{e}_i}/n = 0$, so the test is essentially based on $\sum_{i=1}^{n}{\hat{Z}_i' \tilde{e}_i}/n$.

Let $\tilde{Z} = M_{\hat{W}} \hat{Z}$. The LM type test statistic is given by:
\begin{align}\label{xi_hc}
\xi_{HC} &= \left( \sum_{i=1}^{n}{\tilde{e}_i' \tilde{Z}_i} \right) \left( \sum_{i=1}^{n}{\tilde{Z}_i' \tilde{e}_i \tilde{e}_i' \tilde{Z}_i} \right)^{-1} \left( \sum_{i=1}^{n}{\tilde{Z}_i' \tilde{e}_i} \right) 
\end{align}

Alternatively, in the homoskedastic case, it can be simplified as follows:
\begin{align}\label{xi}
\xi = \left( \sum_{i=1}^{n}{\tilde{e}_i' \tilde{Z}_i} \right) \left( \sum_{i=1}^{n}{\tilde{Z}_i' \tilde{\Sigma}_T \tilde{Z}_i} \right)^{-1} \left( \sum_{i=1}^{n}{\tilde{Z}_i' \tilde{e}_i} \right),
\end{align}
where $\tilde{\Sigma}_T = \frac{1}{n} \sum_{i=1}^{n}{\tilde{e}_i \tilde{e}_i'}$.

These two test statistics resemble the parametric LM test statistic. However, the number of restrictions $r_n$ is allowed to grow to infinity. Thus, in order to obtain convergence in distribution, a normalization is needed. The normalized test statistics are given by
\begin{align}\label{t_test_statistics}
t_{HC} = \frac{\xi_{HC} - r_n}{\sqrt{2 r_n}} \quad \text{and} \quad t = \frac{\xi - r_n}{\sqrt{2 r_n}}
\end{align}

I will show in the next section that under appropriate conditions, the normalized test statistics are asymptotically standard normal. 

\section{Asymptotic Theory}\label{asymptotics}

In this section, I develop the asymptotic theory for the proposed specification test. I analyze its behavior under the null hypothesis and under a fixed alternative.

\subsection{Behavior of the Test Statistic under $H_0$}

This section derives the asymptotic distribution of the test statistic when the semiparametric model is correctly specified. I start with my assumptions. First, I impose some regularity conditions on the data generating process.

\begin{assumption}\label{dgp}

$(Y_{it}, X_{it}')' \in \mathbb{R}^{1+d_x}, d_x \in \mathbb{N}, i = 1, ..., n$ are independent across individuals, i.e. $(Y_i', X_i')'$ are i.i.d. random draws of the random variables $(Y_1', X_1')'$, and the support of $X_1$, $\mathcal{X}$, is a compact subset of $\mathbb{R}^{d_x}$.

\end{assumption}

\begin{assumption}\label{errors_unified}

Let $\varepsilon_i = Y_i - E[Y_i | X_i]$. The following two conditions hold:

\begin{enumerate}[(a)]

\item

$\Sigma(x) = E[\varepsilon_{i} \varepsilon_{i}' | X_i = x]$ is bounded.

\item

$E[\varepsilon_{it}^4 | X_i]$ is bounded.

\end{enumerate}

\end{assumption}

The following assumption deals with the behavior of the approximating series functions. From now on, let $\| A \| = [tr(A'A)]^{1/2}$ be the Euclidian norm of a matrix $A$. Let $x \in R^{d_x}$ be a realization of the random variable $X_{it}$.

\begin{assumption}\label{series_norms_eigenvalues}

For each $m$, $r$, and $k$ there are matrices $B_1$, and $B_2$ such that, for $\bar{W}^{m}(x) = B_1 \hat{W}^{m}(x)$, $\bar{Z}^{r}(x) = B_2 \hat{Z}^{r}(x)$, and $\bar{P}^{k}(x) = (\bar{W}^{m}(x)', \bar{Z}^{r}(x)')$,

\begin{enumerate}[(a)]

\item

There exists a sequence of constants $\zeta(\cdot)$ that satisfies the conditions $\sup_{x \in \mathcal{X}} \| \bar{W}^{m}(x) \| \leq \zeta(m)$, $\sup_{x \in \mathcal{X}} \| \bar{Z}^{r}(x) \| \leq \zeta(r)$, and $\sup_{x \in \mathcal{X}} \| \bar{P}^{k}(x) \| \leq \zeta(k)$.

\item

The smallest eigenvalue of $E[\bar{P}^{k}(X_{it}) \bar{P}^{k}(X_{it})']$ is bounded away from zero uniformly in $k$.

\end{enumerate}

\end{assumption}

\begin{assumption}\label{series_approx}
Suppose that $H_0$ holds. There exist $\alpha > 0$ and $\beta_1 \in \mathbb{R}^{m_n}$ such that
\[
\sup_{x \in \mathcal{X}}{| f(x,\theta_0,h_0) - W^{m_n}(x)' \beta_1 |} = O(m_n^{-\alpha})
\]
\end{assumption}

$\beta_1$ in this assumption can be defined in various ways. One natural definition is projection: $\beta_1 = E[\hat{W}_{it} \hat{W}_{it}']^{-1} E[\hat{W}_{it} \hat{f}(X_{i},\theta_0,h_0)]$, where $\hat{f}(\cdot)$ is an appropriate transformation of $f(\cdot)$.

\begin{theorem}\label{asy_distr_t_r_n_hc}
Assume that Assumptions \ref{dgp}, \ref{errors_unified}, \ref{series_norms_eigenvalues}, and \ref{series_approx} are satisfied, and the following rate conditions hold:
\begin{align}
\label{rate_cond_r_1_hc} (m_n/n + m_n^{-2\alpha}) \zeta(r_n)^2 r_n^{1/2} &\to 0 \\
\label{rate_cond_r_2_hc} \zeta(r_n) r_n / n^{1/2} &\to 0 \\
\label{rate_cond_r_3_hc} \zeta(k_n) m_n^{1/2} k_n^{1/2}/n^{1/2} &\to 0 \\
\label{rate_cond_r_4_hc} n m_n^{-2\alpha}/ r_n^{1/2} &\to 0 \\
\label{rate_cond_r_5_hc} \zeta(r_n)^2/n^{1/2} &\to 0
\end{align}

Also assume that $\| \hat{\Omega} - \tilde{\Omega} \| = o_p(r_n^{-1/2})$, where
\[
\tilde{\Omega} = n^{-1} \sum_{i=1}^{n}{\hat{Z}_i' \tilde{e}_i \tilde{e}_i' \hat{Z}_i} \text{ and } \hat{\Omega} = n^{-1}\sum_{i=1}^{n}{\tilde{Z}_i' \tilde{e}_i \tilde{e}_i' \tilde{Z}_i}.
\]
Then under $H_0$
\begin{align}\label{eqn_t_r_n_HC1}
t_{HC} = \frac{\xi_{HC} - r_n}{\sqrt{2 r_n}} \overset{d}{\to} N(0,1),
\end{align}
where $\xi_{HC}$ is as in Equation~\ref{xi_hc}.

If, in addition to the assumptions above, $\Sigma(x) \equiv \Sigma$ and $\| \hat{\Omega} - \tilde{\Omega} \| = o_p(r_n^{-1/2})$, where
\[
\tilde{\Omega} = n^{-1} \sum_{i=1}^{n}{\hat{Z}_i' \tilde{\Sigma}_T \hat{Z}_i} \text{ and } \hat{\Omega} = n^{-1} \sum_{i=1}^{n}{\tilde{Z}_i' \tilde{\Sigma}_T \tilde{Z}_i},
\]
then
\[
t = \frac{\xi - r_n}{\sqrt{2 r_n}} \overset{d}{\to} N(0,1),
\]
where $\xi$ is as in Equation~\ref{xi}.
\end{theorem}

The normalization I use, $r_n$, differs from the normalization used in most series-based specification tests for parametric models with cross-sectional data, which use the total number of parameters in the nonparametric model $k_n$ (see equations (2.1) and (2.2) in \citet{hong_white_1995} and Lemma 6.2 in \citet{donald_et_al_2003}). This difference can be viewed as a degrees of freedom correction.

The fact that I am dealing with semiparametric, as opposed to parametric, models requires me to modify the key step of my proof, going from the transformed semiparametric regression residuals $\tilde{e}$ to the transformed true errors $\hat{\varepsilon}$. My approach relies on the projection property of series estimators to eliminate the estimation variance and hence only needs to deal with the approximation bias. Specifically, it uses the equality $\tilde{e} = M_W \hat{\varepsilon} + M_W \hat{R}$, applies a central limit theorem for $U$-statistics to the quadratic form in $M_W \hat{\varepsilon}$, and bounds the remainder terms by requiring the approximation error $R$ to be small.

The conventional approach does not impose any special structure on the model residuals and uses the equality $\tilde{e} = \hat{\varepsilon} + (\hat{g} - \tilde{g})$. In parametric models, $\hat{g} - \tilde{g} = \hat{X}' (\beta - \hat{\beta})$, and $\hat{\beta}$ is $\sqrt{n}$-consistent. This makes it possible to apply a central limit theorem for $U$-statistics to the quadratic form in $\hat{\varepsilon}$ and bound the remainder terms that depend on $\hat{X}' (\beta - \hat{\beta})$. However, in semiparametric models this approach needs to deal with both the bias and variance of semiparametric estimators. Specifically, $\hat{g} - \tilde{g} = \hat{R} + \hat{W}'(\beta_1 - \tilde{\beta}_1)$, where $\hat{R}$ can be viewed as the bias term and $\hat{W}'(\beta_1 - \tilde{\beta}_1)$ as the variance term. Thus, in order for $(\hat{g} - \tilde{g})$ to be small, both bias and variance need to vanish sufficiently fast, and the resulting rate conditions turn out to be very restrictive. To see this, it is useful to look at the rates that would be permissible with and without the degrees of freedom correction.

Usually $\zeta(k) = O(k^{1/2})$ for splines and $\zeta(k) = O(k)$ for power series. It can be shown that if splines are used, the rates $k_n = O(n^{2/7})$, $r_n = O(n^{2/7})$, $m_n = O(n^{1/4})$ are permissible if $\alpha \geq 4$. If power series are used, the rates $k_n = O(n^{2/9})$, $r_n = O(n^{2/9})$, $m_n = O(n^{1/5})$ are permissible if $\alpha \geq 5$.

Without the degrees of freedom correction, in order for the test to be asymptotically valid, $m_n$ typically has to be of the order $o(k_n^{1/2})$. Hence, $k_n = O(n^{2/7})$ would require $m_n = o(n^{1/7})$ and $\alpha \geq 7$ if splines are used. If power series are used, $k_n = O(n^{2/9})$ would require $m_n = o(n^{1/9})$ and $\alpha \geq 9$.

\subsection{A Wild Bootstrap Procedure}\label{wild_bootstrap}

In this section I propose a wild bootstrap procedure that can be used to obtain critical values for my test and establish its asymptotic validity. I will compare the small sample behavior of the asymptotic and bootstrap versions of the test in simulations.

Because I am interested in approximating the asymptotic distribution of the test under the null hypothesis, the bootstrap data generating process should satisfy the null. Moreover, because my test is robust to heteroskedasticity, the bootstrap data generating process should be able to accommodate heteroskedastic errors. Finally, because the errors in panel data models may be correlated over time (but not across units), the bootstrap procedure should take this into account. The wild bootstrap can satisfy both these requirements.

The bootstrap procedure will be based on the residuals based on the transformed data $\tilde{e}_{it} = \hat{Y}_{it} - \hat{W}_{it}' \tilde{\beta}_1$. I require the bootstrap errors to satisfy the following two requirements:
\[
\text{(i) } E^*[\hat{\varepsilon}_i^*] = 0, \quad \text{(ii) } E^*[\hat{\varepsilon}_i^{*} \hat{\varepsilon}_i^{* \prime}] = \tilde{e}_i \tilde{e}_i',
\]
where $E^*[\cdot] = E[\cdot | \mathcal{Z}_{n,T}]$ is the expectation conditional on the data $\mathcal{Z}_{n,T} = \{ (Y_{it}, X_{it}')' \}_{i=1,t=1}^{n,T}$. To satisfy these requirements, I let $\hat{\varepsilon}_i^* = V_i^* \tilde{e}_i$, where $V_i^*$ is a two-point distribution. Note that I use the same $V_i^*$ for all time periods for a given $i$. By doing so, I maintain the intertemporal correlation of the transformed errors and residuals in the original sample.

Various choices of $V_i^*$ are possible. One popular option is Mammen's two point distribution, originally introduced in \citet{mammen_1993}: 
\[ 
V_i^* =
  \begin{cases}
    (1-\sqrt{5})/2       & \quad \text{with probability } (\sqrt{5}+1)/(2 \sqrt{5}),\\
    (1+\sqrt{5})/2      & \quad \text{with probability } (\sqrt{5}-1)/(2 \sqrt{5}).
  \end{cases}
\]

Another possible choice is the Rademacher distribution, as suggested in \citet{davidson_flachaire_2008}:
\[ 
V_i^* =
  \begin{cases}
    -1       & \quad \text{with probability } \frac{1}{2}, \\
    1      & \quad \text{with probability } \frac{1}{2}.
  \end{cases}
\]

The wild bootstrap procedure then works as follows:

\begin{enumerate}

\item

Obtain the estimates $\tilde{\beta}_1$ and residuals $\tilde{e}_i$ from the restricted model $\hat{Y}_{it} = \hat{W}_{it}' \beta_1 + \hat{e}_{it}$.

\item

Generate the wild bootstrap error $\hat{\varepsilon}_i^* = V_i^* \tilde{e}_i$.

\item

Obtain $\hat{Y}_{it}^* = \hat{W}_{it}' \tilde{\beta}_1 + \hat{\varepsilon}_{it}^*$. Then estimate the restricted model and obtain the restricted bootstrap residuals $\tilde{e}_{it}^*$ using the bootstrap sample $\{ (\hat{Y}_{it}^*, \hat{W}_{it}')' \}_{i=1,t=1}^{n,T}$.

\item

Use $\tilde{e}_{it}^*$ in place of $\tilde{e}_{it}$ to compute the bootstrap test statistic $t_{HC,r_n}^*$ or  $t_{r_n}^*$.

\item

Repeat steps 2--4 $B$ times (e.g. $B=399$) and obtain the empirical distribution of the $B$ test statistics $t_{r_n}^*$ or $t_{HC,r_n}^*$. Use this empirical distribution to compute the bootstrap critical values of the bootstrap $p$-values.

\end{enumerate}

Then the following is true.

\begin{theorem}\label{asy_distr_t_r_n_hc_boot}
Assume that Assumptions of Theorem~\ref{asy_distr_t_r_n_hc} hold. Let $\mathcal{Z}_{n,T} = \{ (Y_{it}, X_{it}')' \}_{i=1,t=1}^{n,T}$. Then
\[
F_{HC,n}^*(t) \to \Phi(t) \text{ in probability},
\]
for all $t$, as $n \to \infty$, where $F_{HC,n}^*(t)$ is the bootstrap distribution of $t_{HC,r_n}^*|\mathcal{Z}_{n,T}$ and $\Phi(\cdot)$ is the standard normal CDF.
\end{theorem}

A similar result can be obtained for the homoskedastic test statistic $t_{r_n}^*$. It is omitted for brevity.

\subsection{Behavior of the Test Statistic under a Fixed Alternative}

This section discusses the behavior of the test statistic under a fixed alternative. First, a cautionary note is in order. Note that the null hypothesis concerns the model
\[
Y_{it} = g(X_{it}) + \mu_i + \varepsilon_{it},
\]
while the semiparametric series estimation method is based on the transformed model
\[
\hat{Y}_{it} = \hat{g}(X_{i}) + \hat{\varepsilon}_{it},
\]
where the model is transformed by taking the first differences or using the within transformation.\footnote{Because of this, $\hat{g}(\cdot)$ may depend on all elements of $X_i$, not just $X_{it}$.} In particular, the model is estimated based on the series form
\[
\hat{Y}_{it} = \hat{W}_{it}' \beta_1 + e_{it}
\]

Because of this, my test will only be able to detect specification errors that are present in the transformed model. In other words, the null hypothesis essentially becomes $H_0: P(\hat{g}(X_i) = \hat{f}(X_i,\theta_0,h_0)) = 1$ for some $\theta_0$ and $h_0$. Because most of the time researchers work with transformed models when they deal with fixed effects, I believe this is a reasonable hypothesis to test. Other specification tests for fixed effects panel data models, e.g. in \citet{lin_et_al_2014}, are also usually based on the transformed residuals. As long as the transformation used does not eliminate the specification error in the original model, the test will be consistent for the original model.

\begin{assumption}\label{assumption_din_1}
(\citet{donald_et_al_2003}, Assumption 1) 

Assume that $E[\hat{P}_{it} \hat{P}_{it}']$ is finite for all $k$, and for any $a(x)$ with $E[a(X_{i})^2]<\infty$ there are $k \times 1$ vectors $\gamma^{k}$ such that, as $k \to \infty$,
\[
E[(a(X_{i}) - \hat{P}_{it}' \gamma^{k})^2] \to 0
\]
\end{assumption}

Lemma~\ref{lemma_din_1} in the appendix shows that when this assumption is satisfied, the conditional moment restriction $E[\varepsilon_i | X_i] = 0$ is equivalent to a growing number of unconditional moment restrictions. The class of functions $a(x)$, for which the equivalence between the conditional and unconditional restrictions holds, consists of functions that can be approximated (in the mean squared sense) using series as the number of series terms grows. While it is difficult to give a necessary and sufficient primitive condition that would describe this class of functions, the test will likely be consistent against continuous and smooth alternatives, while it may not be consistent against alternatives that exhibit jumps.

This is a population result in the sense that it does not involve the sample size $n$. In order to use this result in practice, I require the number of series terms used to construct the test statistic, $k_n$, to grow with the sample size. By doing so, I ensure that the unconditional moment restriction $E[\hat{P}_i' \hat{\varepsilon}_i] = 0$, on which the test is based, is equivalent to the conditional moment restriction $E[\hat{\varepsilon}_i | X_i]$. Thus, the test will be consistent against a wide class of alternatives satisfying Assumption~\ref{assumption_din_1}.

In order to analyze the behavior of the test under a fixed alternative, I introduce some notation first. The true model is nonparamertic:
\[
Y_{it} = g(X_{it}) + \mu_i + \varepsilon_{it}, \quad E[\varepsilon_i | X_i] = 0
\]

An alternative way to write this model is
\[
Y_{it} = f(X_{it},\theta^*,h^*) + \mu_i + \varepsilon_{it}^*,
\]
where $\theta^*$ and $h^*$ are pseudo-true parameter values and $\varepsilon_{it}^* = \varepsilon_{it} + (g(X_{it}) - f(X_{it},\theta^*,h^*)) = \varepsilon_{ti} + d(X_{it})$ is a composite error term. The pseudo-true parameter values minimize
\[
E[(g(X_{it}) - f(X_{it},\theta,h))^2]
\]
over a suitable parameter space.

Note that the model can be written as
\[
Y_{it} = W_{it}' \beta_1^{*} + \mu_i + \varepsilon_{it}^* + R_{it}^*,
\]
where $R_{it}^* = (f(X_{it},\theta^*,h^*) - W_{it}' \beta_1^{*})$. After transforming the data, the model becomes
\[
\hat{Y}_{it} = \hat{W}_{it}' \beta_1^* + \hat{\varepsilon}_{it}^* + \hat{R}_{it}^*
\]

The pseudo-true parameter value $\beta_1^*$ solves the moment condition $E[\hat{W}_{it} (\hat{Y}_{it} - \hat{W}_{it}' \beta_1^{*})] = 0$, and the semiparametric estimator $\tilde{\beta}_1$ solves its sample analog $\hat{W}' (\hat{Y} - \hat{W} \tilde{\beta}_1)/n = 0$.

The following theorem provides the divergence rate of the test statistic under the fixed alternative.

\begin{theorem}\label{global_alternative_t_r}
Let $\Omega^* = E[\hat{Z}_i' \hat{\varepsilon}_i \hat{\varepsilon}_i' \hat{Z}_i]$. In the heteroskedastic case, let
\[
\hat{\Omega} = n^{-1}\sum_{i=1}^{n}{\tilde{Z}_i' \tilde{e}_i \tilde{e}_i' \tilde{Z}_i},
\]
and in the homoskedastic case let
\[
\hat{\Omega} = n^{-1} \sum_{i=1}^{n}{\tilde{Z}_i' \tilde{\Sigma}_T \tilde{Z}_i}
\]

Suppose that there exists $\beta_1^*$ such that $\sup_{x \in \mathcal{X}}{| f(x,\theta^*,h^*) - W^{m_n}(x)' \beta_1^* |} \to 0$, $\| \hat{\Omega} - \Omega^* \| \overset{p}{\to} 0$, the smallest eigenvalue of $\Omega^*$ is bounded away from zero, $m_n \to \infty$, $r_n \to \infty$, $r_n/n \to 0$, $E[\hat{\varepsilon}_i^{*\prime} T_i] \Omega^{*-1} E[T_i' \hat{\varepsilon}_i^*] \to \Delta$, where $\Delta$ is a constant. Then under homoskedasticity
\[
\frac{\sqrt{r_n}}{n} \frac{\xi - r_n}{\sqrt{2 r_n}} \overset{p}{\to} \Delta/\sqrt{2},
\]
and under heteroskedasticity
\[
\frac{\sqrt{r_n}}{n} \frac{\xi_{HC} - r_n}{\sqrt{2 r_n}} \overset{p}{\to} \Delta/\sqrt{2}
\]
\end{theorem}


\section{Simulations}\label{simulations}

In this section, I study the finite sample performance of the proposed test using simulations. I have several goals: first, to illustrate the importance of the degrees of freedom correction; second, to study the sensitivity of the test to the choice of basis functions and tuning parameters; third, to compare the asymptotic version of my test with its bootstrap version; finally, to study the effect of the sample size on the test behavior.

The setup I use resembles the one in \citet{korolev_2019} but includes fixed effects:
\[
Y_{it} = \mu_i +  X_{1it} \beta + g(X_{2it}) + \varepsilon_{it}
\]
Here $\varepsilon_{it}$ are independent across individuals $i$ and time $t$, while $\alpha_i$ are fixed effects that are correlated with both the regressors and error terms for individual $i$. More specifically,
\[
\mu_i = \nu_i + \mu_{X,i},
\]
where $\nu_i \sim \text{i.i.d. } N(0,2.25)$ and $\mu_{X,i} =  \sum_{t=1}^{T}{(0.6 X_{1it}+0.4 X_{2it})}$. In this setting, estimating the model $Y_{it} = 2 X_{1it} + g(X_{2it}) + e_{it}$, where $e_{it} = \mu_i + \varepsilon_{it}$, would result in inconsistent estimates, so it is crucial to account for the panel nature of the data and for the presence of fixed effects. To achieve this, I use the within transformation.\footnote{I have tried using first differencing instead of the within transformation and obtained similar results.} After that, I estimate the model and compute the proposed test statistic.

I test the following null hypothesis:
\begin{align*}
H_0^{SP}: P \left(E[Y_{it} | \mu_i, X_{it}] = \mu_i +  X_{1it} \beta + g(X_{2it}) \right) = 1 \text{ for some } \beta, g(X_{2it})
\end{align*}
against the alternative
\begin{align*}
H_1: P \left(E[Y_{it} |  \mu_i, X_{it}] \neq \mu_i +  X_{1it} \beta + g(X_{2it}) \right) > 0 \text{ for all } \beta, g(X_{2it})
\end{align*}

I use two data generating processes:

\begin{enumerate}

\item

Semiparametric partially linear, which corresponds to $H_0^{SP}$:
\begin{align}\label{DGP_SP}
\begin{split}
Y_{it} &= \mu_i + 2 X_{1it} + g(X_{2it}) + \varepsilon_{it} \\
g(X_{2it}) &= 3 + 2(\exp(X_{2it})-2 \ln(X_{2it}+3))
\end{split}
\end{align}

\item

Nonparametric, which corresponds to $H_1$:
\begin{align}\label{DGP_NP}
\begin{split}
Y_{it} &= \mu_i + 2 X_{1it} + g(X_{2it}) + h(X_{1it},X_{2it}) +  \varepsilon_{it} \\
h(X_{1it},X_{2it}) &= h_1(X_{1it}) h_2(X_{2it}) \\
h_1(X_{1it}) &= 1.25 \cos(X_{1it}-2), \quad h_2(X_{2it}) = \sin(0.75 X_{2it})
\end{split}
\end{align}

\end{enumerate}

These two DGPs are very similar to the ones used in \citet{korolev_2019}, but include fixed effects. For more details on the two DGPs, see \citet{korolev_2019}. I consider four setups: Setup 1 with $(n = 250, T = 2)$, Setup 2 with $(n = 250, T = 4)$, Setup 3 with $(n = 500, T = 2)$, Setup 4 with $(n = 500, T = 4)$. I separately consider two settings: with homoskedastic errors and with heteroskedastic errors.

To implement the test, I use both power series and cubic splines as basis functions due to their popularity. Instead of studying the behavior of my test for a given (arbitrary) number of series terms, I vary the number of terms in univariate series expansions to investigate how the behavior of the test changes as a result. The total number of parameters $k_n$ ranges from 15 to 39 in Setups 1 and 2 and to 52 in Setups 3 and 4.\footnote{For more details, see the online supplement to \citet{korolev_2019}.}

\subsection{Homoskedastic Errors}\label{simulations_homoskedastic_errors}

First, I investigate the performance of the test when the errors are homoskedastic. The errors are normally distributed and independent across both $i$ and $t$: $\varepsilon_{it} \sim \text{i.i.d. } N(0,4)$. I consider tests based both on the LM type test statistic
\[
\xi = \left( \sum_{i=1}^{n}{\tilde{e}_i' \tilde{Z}_i} \right) \left( \sum_{i=1}^{n}{\tilde{Z}_i' \tilde{\Sigma}_T \tilde{Z}_i} \right)^{-1} \left( \sum_{i=1}^{n}{\tilde{Z}_i' \tilde{e}_i} \right) \overset{a}{\sim} \chi^2(\tau_n)
\]
and on the normalized statistic $t_{\tau_n} = \frac{\xi - \tau_n}{\sqrt{2 \tau_n}} \overset{a}{\sim} N(0,1)$.

I start by looking at the simulated size of the test at the nominal 5\% level. Figures~\ref{fig_simulated_size_set1}, \ref{fig_simulated_size_set2}, \ref{fig_simulated_size_set3}, and \ref{fig_simulated_size_set4} plot the simulated size as a function of the number of series terms in univariate series expansions $a_n$ (including the constant term) for the four setups I consider. The upper panels of these figures use the LM type statistic $\xi$ , while the bottom panels use the normalized test statistic $t$. The left panels use power series and the right panels use splines. I consider four versions of the test: the asymptotic version with $\tau_n = r_n$ (red solid lines), the asymptotic version with $\tau_n = k_n$ (magenta solid lines), the wild bootstrap version with the Rademacher distribution (cyan dash-dotted lines), and the wild bootstrap with Mammen's distribution (blue dashed lines).

As we can see, the asymptotic test without the degrees of freedom correction (i.e. with $\tau_n = k_n$) is severely undersized. In turn, the asymptotic test with the degrees of freedom correction (i.e. with $\tau_n = r_n$) based on the $t$ statistic is slightly oversized, while the asymptotic test based on the $\xi$ statistic controls size very well. Depending on the setup, its performance is either very close to, or even better than, that of the wild bootstrap tests. We can also see that the performance of the test is fairly robust to the choice of basis functions and tuning parameters.

Next, I turn to the test power. Figures~\ref{fig_simulated_power_set1}, \ref{fig_simulated_power_set2}, \ref{fig_simulated_power_set3}, and \ref{fig_simulated_power_set4} plot the simulated power of the nominal 5\% level test as a function of the number of series terms in univariate series expansions $a_n$. Given that the asymptotic test without the degrees of freedom correction is undersized, it is not surprising that it also has very low power. In turn, the power of the asymptotic version of the test with the degrees of freedom correction is very similar to the power of the wild bootstrap tests. As could be expected, the power increases as the sample size (the number of units $n$ or the number of periods $T$) increases. Finally, the power decreases as the number of series terms grows. This is due to the fact that the alternative is smooth and can be captured by the first few series terms. I will turn to a data-driven method to choose tuning parameters later.

\subsection{Heteroskedastic Errors}\label{simulations_heteroskedastic_errors}

In this section I investigate the performance of the test when the errors are heteroskedastic. The errors are normally distributed and independent across both $i$ and $t$, but not identically distributed: $\varepsilon_{it} \sim \text{i.n.i.d. } N(0,1+1.75 \exp(0.75(X_{1it}+X_{2it})))$. I consider tests based both on the heteroskedasticity-robust LM type test statistic
\[
\xi_{HC} = \left( \sum_{i=1}^{n}{\tilde{e}_i' \tilde{Z}_i} \right) \left( \sum_{i=1}^{n}{\tilde{Z}_i' \tilde{e}_i \tilde{e}_i' \tilde{Z}_i} \right)^{-1} \left( \sum_{i=1}^{n}{\tilde{Z}_i' \tilde{e}_i} \right) 
\]
and on the normalized statistic $t_{\tau_n,HC} = \frac{\xi_{HC} - \tau_n}{\sqrt{2 \tau_n}} \overset{a}{\sim} N(0,1)$.

First, I look at  the simulated size of the test at the nominal 5\% level. Figures~\ref{fig_simulated_size_set1_hc}, \ref{fig_simulated_size_set2_hc}, \ref{fig_simulated_size_set3_hc}, and \ref{fig_simulated_size_set4_hc} plot the simulated size as a function of the number of series terms in univariate series expansions $a_n$. We can see that the asymptotic test without the degrees of freedom correction is again severely undersized. The asymptotic test with the degrees of freedom correction based on the $\xi_{HC}$ statistic is also undersized, though its size becomes closer to the nominal level as the sample size grows. In turn, the simulated size of the asymptotic test with the degrees of freedom correction based on the $t_{HC}$ statistic is pretty close to the nominal level. In fact, in Setups 1 and 2, when splines are used, it controls size even better that the wild bootstrap tests.

Next, I turn to the test power. Figures~\ref{fig_simulated_power_set1_hc}, \ref{fig_simulated_power_set2_hc}, \ref{fig_simulated_power_set3_hc}, and \ref{fig_simulated_power_set4_hc} plot the simulated power of the nominal 5\% level test as a function of the number of series terms in univariate series expansions $a_n$. The asymptotic test without the degrees of freedom correction has low power in all setups. The asymptotic test with the degrees of freedom correction based on the $\xi_{HC}$ test statistic is less powerful than the wild bootstrap tests, but the power loss decreases as the sample size grows. Finally, the power of the asymptotic test with the degrees of freedom correction based on the $t_{HC}$ statistic is fairly close to the power of the bootstrap tests, especially with larger sample sizes.

To summarize, even though the performance of the asymptotic test with the degrees of freedom correction deteriorates when the errors are heteroskedastic, as opposed to homoskedastic, it nevertheless comes close to the wild bootstrap tests in most setups. With larger sample sizes, its performance is almost indistinguishable from that of the bootstrap tests.

\subsection{Data-Driven Methods for Tuning Parameters Choice}\label{simulations_data_driven}

In the simulations presented above, the alternative was smooth and the power of the test declined as the number of series terms increased. However, this is not always the case. There exist alternatives that are orthogonal to the first few series terms, and in order to detect such alternatives, one needs to include higher order series terms. In this section, I investigate the finite sample performance of a data-driven method to select tuning parameters.

In order to simplify the problem, I abstract away from the task of selecting the number of series terms under the null and consider a linear univariate null model:
\[
Y_{it} = \mu_i + 2 X_{1it} + \varepsilon_{it}
\]

The smooth alternative is given by
\[
Y_{it} = \mu_i + 2 X_{1it} + \cos(X_{1it} - 2) + \varepsilon_{it}
\]
I also consider an alternative that is orthogonal to the first four power terms in $X_1$. A data-driven test should be able to adapt to a wide class of alternatives and choose tuning parameters appropriately.

I use a modified version of the approach proposed in \citet{guay_guerre_2006}. I use the $\xi$ test statistic and pick the value of $r_n$ that maximizes
\[
\xi(r_n) - r_n - \gamma_n \sqrt{2(r_n - r_{n,min})},
\]
where $\gamma_n = c \sqrt{2 \ln{\Card{r_n}}}$, $c$ is a constant that satisfies $c \geq 1 + \varepsilon$ for some $\varepsilon > 0$, $\Card{r_n}$ is the cardinality of the set of possible numbers of restrictions, and $r_{n,min}$ is the lowest possible number of restrictions across different choices of $r_n$. The notation $\xi(r_n)$ emphasizes the dependence of the test statistic $\xi = \left( \sum_{i=1}^{n}{\tilde{e}_i' \tilde{Z}_i} \right) \left( \sum_{i=1}^{n}{\tilde{Z}_i' \tilde{\Sigma}_T \tilde{Z}_i} \right)^{-1} \left( \sum_{i=1}^{n}{\tilde{Z}_i' \tilde{e}_i} \right)$ on the number $r_n$ of elements in $\tilde{Z}$. Intuitively, $r_n$ is the center term of $\xi(r_n)$, while $\gamma_n \sqrt{2(r_n - r_{n,min})}$ is the penalty term that rewards simpler alternatives. In my analysis, I set $c=5$.

Table~\ref{Tbl_data_driven} presents the results. I report the simulated size, power against the standard alternative, and power against the orthogonal alternative for the data driven test and the test with the fixed number of series term equal to $a_n = 4$ and $a_n = 9$ (including the constant term). The former choice of $a_n$ is typically optimal under the regular alternative but has no power against the orthogonal alternative. The latter choice of $a_n$ typically leads to good power against the orthogonal alternative but results in the loss of power against the well-behaved alternative.

As we can see, the data-driven test is slightly oversized in the first three setups and is slightly undersized in the last setup, but overall its size is close to the nominal level. Moreover, it has excellent power against the standard alternative and pretty good power against the orthogonal alternative. Even though a more careful investigation of data-driven specification tests for panel data models is beyond the scope of this paper, my simulations suggest that the proposed procedure performs well in finite samples.

\section{Empirical Example}\label{empirical_example}

In this section, I apply my test to the PSID data\footnote{Available at \url{http://bcs.wiley.com/he-bcs/Books?action=resource&bcsId=4338&itemId=1118672321&resourceId=13452}.} that was used in \citet{cornwell_rupert_1988} and \citet{baltagi_khanti-akom_1990}. The dataset contains 7 years of observations on 595 heads of household between the ages of 18 and 65 in 1976 with a positive reported wage in some private, non-farm employment for all 7 years. Among other models, the authors estimated the following wage equation with fixed effects:
\begin{align}\label{quadratic_model}
\begin{split}
&LWAGE_{it} = \alpha_1 WKS_{it} + \delta_1 EXP_{it} + \lambda_1 EXP_{it}^2 + D_{it}' \gamma_1 + \mu_{i} + \varepsilon_{it} \\
&E[\varepsilon_{i} | WKS_{i}, EXP_{i}, D_{i}, \mu_i] = 0
\end{split}
\end{align}
where $LWAGE_{it}$ is the natural logarithm of the wage of individual $i$ in year $t$, $WKS$ is weeks worked, $EXP$ is experience, and $D_{it}$ includes the following dummy variables: occupation ($OCC= 1$ if the individual has blue-collar occupation), industry ($IND= 1$ if the individual works in a manufacturing industry), residence ($SOUTH = 1$, $SMSA = 1$ if the individual resides in the south, or in a standard metropolitan statistical area), marital status ($MS = 1$ if the individual is married), union coverage ($UNION = 1$ if the individual's wage is set by a union contract).

While my test is fairly general and applies to semiparametric as well as parametric models, I focus on a parametric model because parametric panel data models are prevalent in applications. I test this parametric model against the alternative which is fully nonparametric in weeks worked and experience but is parametric in the dummy variables:
\[
LWAGE_{it} = g(WKS_{it},EXP_{it}) + D_{it}' \gamma + \mu_{i} + \varepsilon_{it},
\]
where $D_{it}$ includes the six dummy variables listed above. Due to the number of dummy variables, considering the alternative which is fully nonparametric appears implausible, as it would essentially require me to split the dataset into $2^{6} = 64$ bins and estimate it within each bin separately.

In order to implement the test, I need to select the basis functions and the number of series terms. I use both power series and splines and utilize a data-driven procedure to select the number of series terms. Note that the number of terms under the null is fixed because the null model is parametric. Thus, I only need to choose the number of series terms under the alternative. Following the approach discussed in Section~\ref{simulations_data_driven}, I vary the number of series terms in univariate series expansions in $WKS$ and $EXP$ from 3 to 8 (not including the constant) and pick the value of $r_n$ that maximizes
\[
\xi_{HC}(r_n) - r_n - \gamma_n \sqrt{2(r_n - r_{n,min})}
\]
I find that the optimal number of terms is equal to 3 (not including the constant term), i.e. that a cubic polynomial should be used. In this case, power series and splines coincide as there are no knots yet. The resulting number of restrictions is $r_n = 12$. The upper panel of Table~\ref{Tbl_testing} reports the heteroskedastic test statistic $\xi_{HC}$ as well as the standardized statistic $t_{HC}$. As we can see, the null hypothesis that the model is correctly specified is not rejected at the 5\% level, but it is rejected at the 10\% level.

Next, I repeat this exercise for the specification that drops the quadratic term in experience. I use the same nonparametric alternative as before. Because the null model has one regressor less than before, I am testing $r_n = 13$ restrictions. The middle panel of Table~\ref{Tbl_testing} reports the results. All for types of the test reject the null hypothesis at any conventional confidence level.

Finally, I estimate a semiparametric model that is nonparametric in experience but is parametric in the remaining variables:
\begin{align}\label{semiparametric_model}
LWAGE_{it} = \alpha_2 WKS_{it} + g_2(EXP_{it}) + D_{it}' \gamma_2 + \mu_{i} + \varepsilon_{it} 
\end{align}

I estimate this semiparametric model using power series with up to cubic terms. The semiparametric model leads to $r_n = 11$ restrictions. Figure~\ref{figure_exp_effects} plots the estimated effects of experience for the linear, quadratic, and emiparametric models. As we can see, the semiparametric model appears to be pretty similar to the quadratic model, and the linear model is not too far off. However, specification testing draws a somewhat different picture. As we can see from the bottom panel of Table~\ref{Tbl_testing}, while the linear model is overwhelmingly rejected and the quadratic model is rejected at the 10\% level, there is no evidence against the semiparametric model.

Because in this paper I develop a specification test and not a model selection procedure, one should be careful with applying my test to several models sequentially. However, it appears that there is substantial evidence against the linear model, while there is little evidence agains the quadratic specification employed by \citet{cornwell_rupert_1988} and \citet{baltagi_khanti-akom_1990}. If the researcher  worries about the borderline results and wants to be on the safe side, it may be plausible to use a more flexible semiparametric model that is fully nonparametric in experience.

\section{Conclusion}\label{conclusion}

In this paper, I develop a Lagrange Multiplier type specification test for semiparametric panel data models with fixed effects. The test achieves consistency by turning a conditional moment restriction into a growing number of unconditional moment restrictions. Unlike in the traditional parametric Lagrange Multiplier test, both the number of parameters and the number of restrictions are allowed to grow with the sample size. I develop an asymptotic theory that explicitly takes this into account and prove that the normalized test statistic converges in distribution to the standard normal.

My test has several attractive features. First, fixed effects panel data models typically require researchers to transform their data, by taking first differences or applying the within transformation. This makes semiparametric estimation and specification testing that involves kernel methods problematic, as it is difficult to impose the additive structure on kernel estimators. In contrast, with series methods, the transformed model remains linear in parameters, and the proposed test is very simple to implement.

Second, the projection property of series estimators allows me to develop a degrees of freedom correction, which explicitly accounts for the variance of semiparametric estimators. Thus, I only need to control the bias, and my rate conditions are relatively mild. Moreover, the degrees of freedom correction results in good performance of the test in simulations.

In future research, I plan to extend the proposed test to semiparametric dynamic panel data models. The presence of endogenous variables calls for the use of instrumental variables. Estimation of such models, with endogeneity only in the parametric part, has been studied in \citet{baltagi_li_2002} and \citet{an_et_al_2016}. A possible concern for specification testing in these models is that nonparametric instrumental variables models are subject to the ill-posed inverse problem, so the unrestricted nonparametric model may not be identified. It remains to be seen whether this identification problem poses a challenge for specification testing in dynamic panel data models.

\clearpage

\renewcommand\thesection{Appendix \Alph{section}}

\renewcommand\thesubsection{\Alph{section}.\arabic{subsection}}

\setcounter{section}{0}

\renewcommand{\thetheorem}{A.\arabic{theorem}}

\renewcommand{\thelemma}{A.\arabic{lemma}}

\renewcommand{\theassumption}{A.\arabic{assumption}}

\renewcommand{\theremark}{A.\arabic{remark}}

\renewcommand{\theequation}{A.\arabic{equation}}

\setcounter{theorem}{0}

\setcounter{lemma}{0}

\setcounter{assumption}{0}

\setcounter{remark}{0}

\section{Tables and Figures}\label{appendix_tables_figures}

\onehalfspacing

\begin{figure}[H]
\begin{center}
\caption{Simulated Size of the Test, $n=250$, $T = 2$}\label{fig_simulated_size_set1}
\includegraphics[scale=0.5]{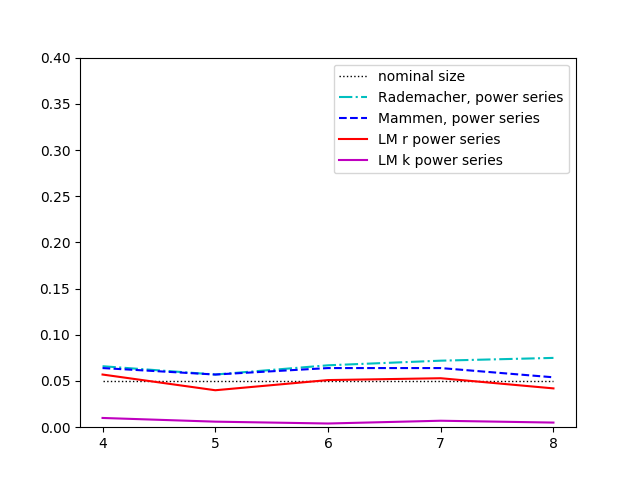} \includegraphics[scale=0.5]{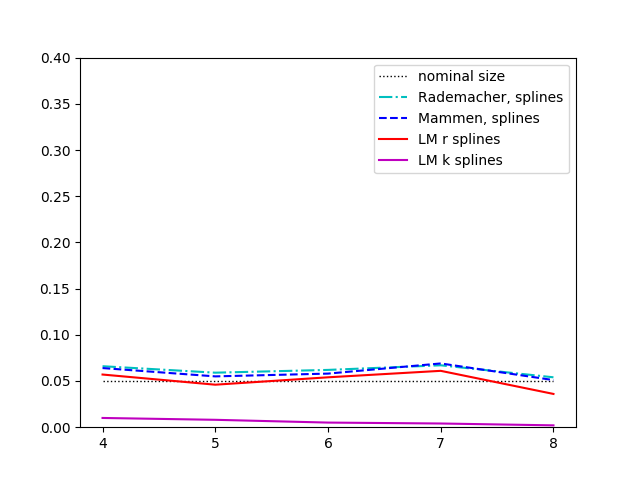}

\includegraphics[scale=0.5]{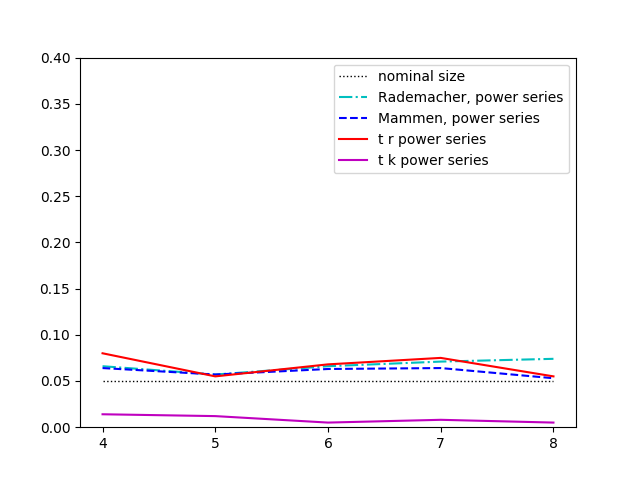} \includegraphics[scale=0.5]{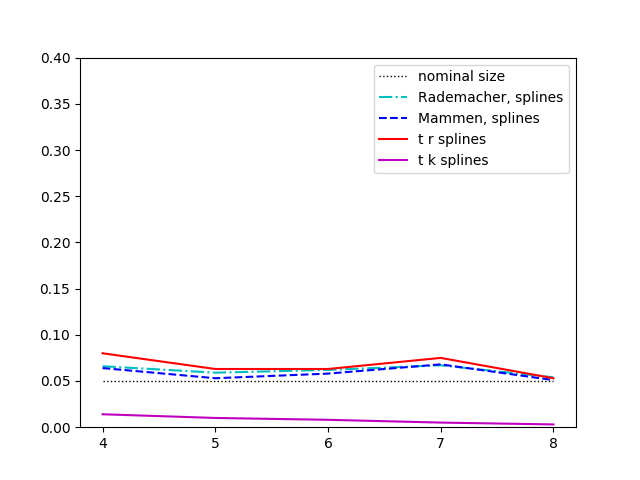}

\end{center}
\footnotesize{\vspace{-0.4cm}This figure plots the simulated size of the nominal 5\% test against the number of series terms in univariate series expansions, $a_n$ (including the constant). The left panel uses power series. The right panel uses splines. The upper panel uses the $\xi$ test statistic from Equation~\ref{xi}. The lower panel uses the $t$ test statistic from Equation~\ref{t_test_statistics}.

The red solid line corresponds to the test that uses the asymptotic critical values and normalization $\tau_n = r_n$. The magenta solid line corresponds to the test that uses the asymptotic critical values and normalization $\tau_n = k_n$. The cyan dash-dotted line corresponds to the test that uses the wild bootstrap critical values based on Rademacher distribution.  The blue dashed line corresponds to the test that uses the wild bootstrap critical values based on Mammen's distribution. The results are based on $M=1,000$ simulations and $B=399$ bootstrap iterations.}
\end{figure}

\begin{figure}[H]
\begin{center}
\caption{Simulated Size of the Test, $n=250$, $T = 4$}\label{fig_simulated_size_set2}
\includegraphics[scale=0.5]{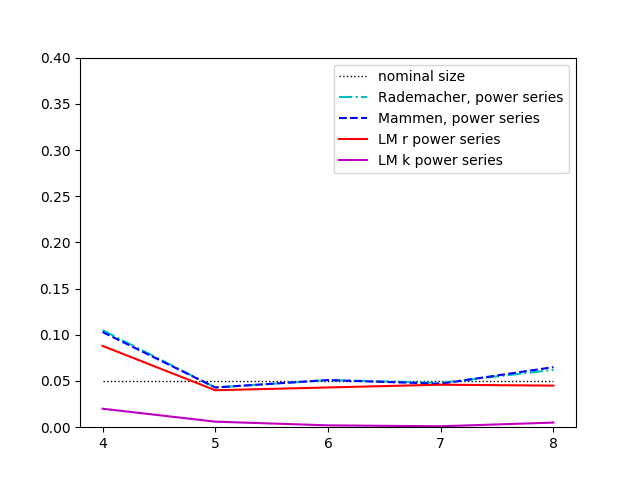} \includegraphics[scale=0.5]{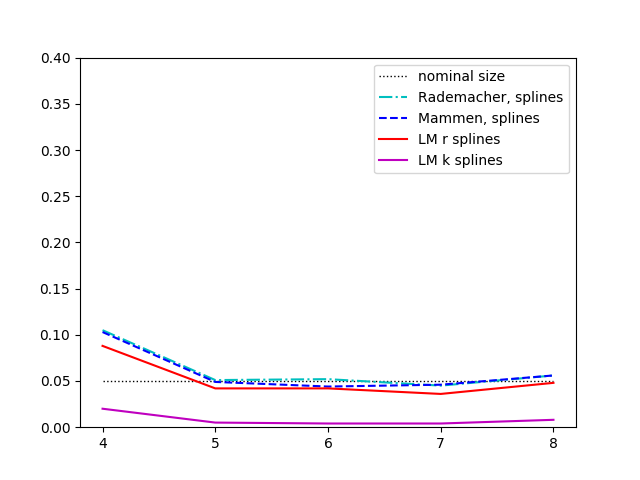}

\includegraphics[scale=0.5]{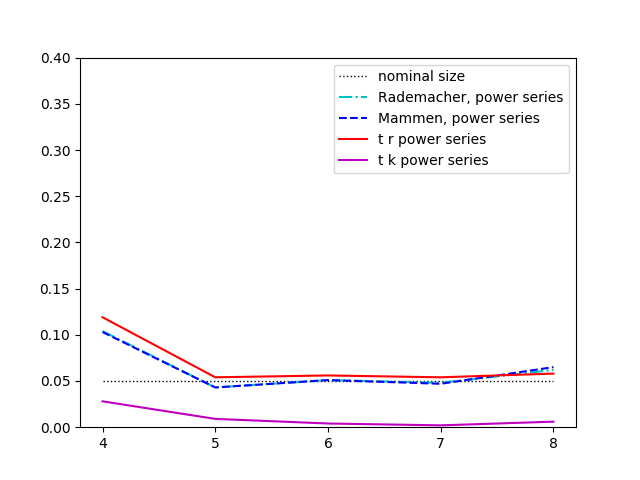} \includegraphics[scale=0.5]{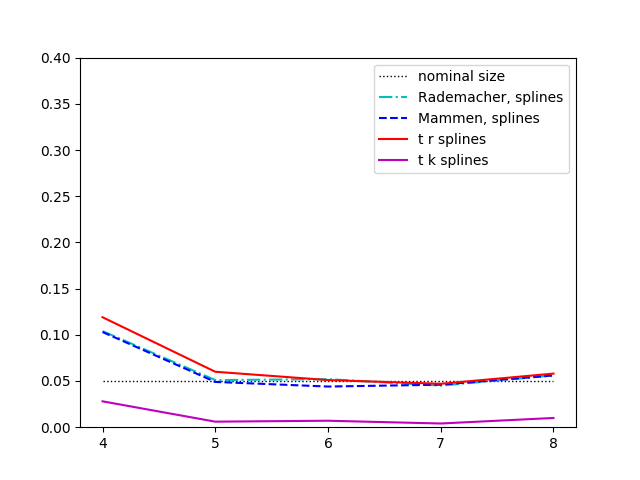}

\end{center}
\footnotesize{\vspace{-0.4cm}This figure plots the simulated size of the nominal 5\% test against the number of series terms in univariate series expansions, $a_n$ (including the constant). The left panel uses power series. The right panel uses splines. The upper panel uses the $\xi$ test statistic from Equation~\ref{xi}. The lower panel uses the $t$ test statistic from Equation~\ref{t_test_statistics}.

The red solid line corresponds to the test that uses the asymptotic critical values and normalization $\tau_n = r_n$. The magenta solid line corresponds to the test that uses the asymptotic critical values and normalization $\tau_n = k_n$. The cyan dash-dotted line corresponds to the test that uses the wild bootstrap critical values based on Rademacher distribution.  The blue dashed line corresponds to the test that uses the wild bootstrap critical values based on Mammen's distribution. The results are based on $M=1,000$ simulations and $B=399$ bootstrap iterations.}
\end{figure}

\begin{figure}[H]
\begin{center}
\caption{Simulated Size of the Test, $n=500$, $T = 2$}\label{fig_simulated_size_set3}
\includegraphics[scale=0.5]{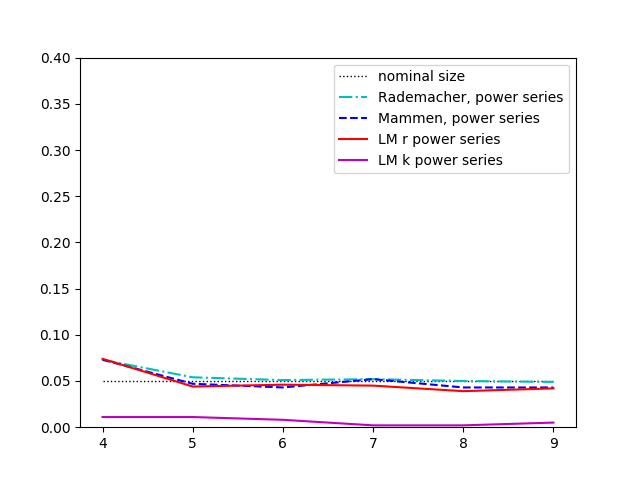} \includegraphics[scale=0.5]{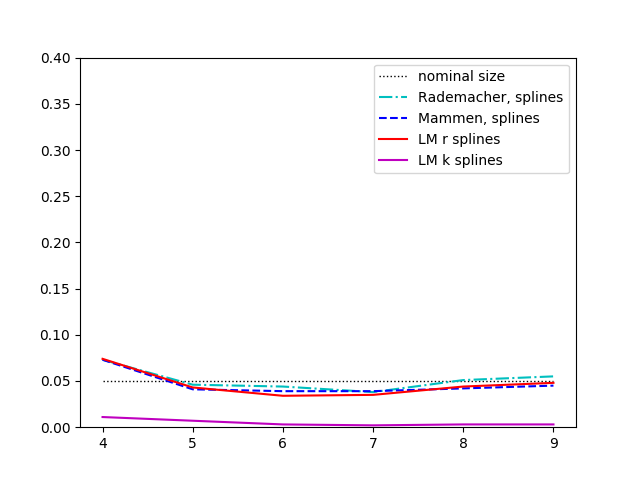}

\includegraphics[scale=0.5]{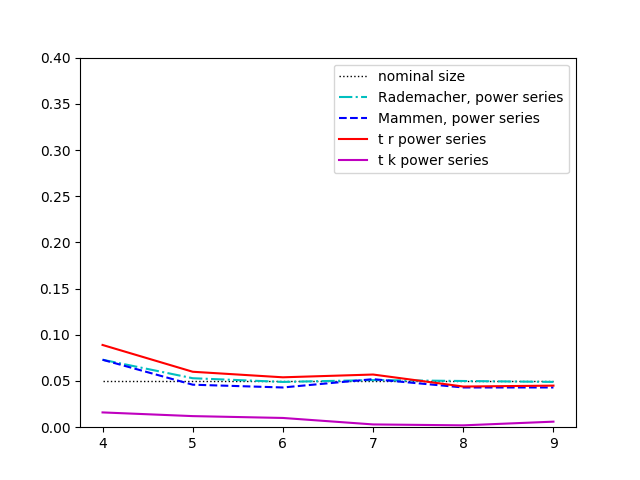} \includegraphics[scale=0.5]{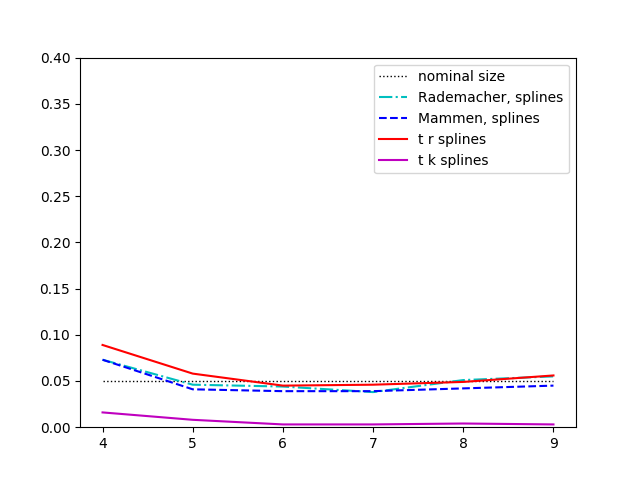}

\end{center}
\footnotesize{\vspace{-0.4cm}This figure plots the simulated size of the nominal 5\% test against the number of series terms in univariate series expansions, $a_n$ (including the constant). The left panel uses power series. The right panel uses splines. The upper panel uses the $\xi$ test statistic from Equation~\ref{xi}. The lower panel uses the $t$ test statistic from Equation~\ref{t_test_statistics}.

The red solid line corresponds to the test that uses the asymptotic critical values and normalization $\tau_n = r_n$. The magenta solid line corresponds to the test that uses the asymptotic critical values and normalization $\tau_n = k_n$. The cyan dash-dotted line corresponds to the test that uses the wild bootstrap critical values based on Rademacher distribution.  The blue dashed line corresponds to the test that uses the wild bootstrap critical values based on Mammen's distribution. The results are based on $M=1,000$ simulations and $B=399$ bootstrap iterations.}
\end{figure}

\begin{figure}[H]
\begin{center}
\caption{Simulated Size of the Test, $n=500$, $T = 4$}\label{fig_simulated_size_set4}
\includegraphics[scale=0.5]{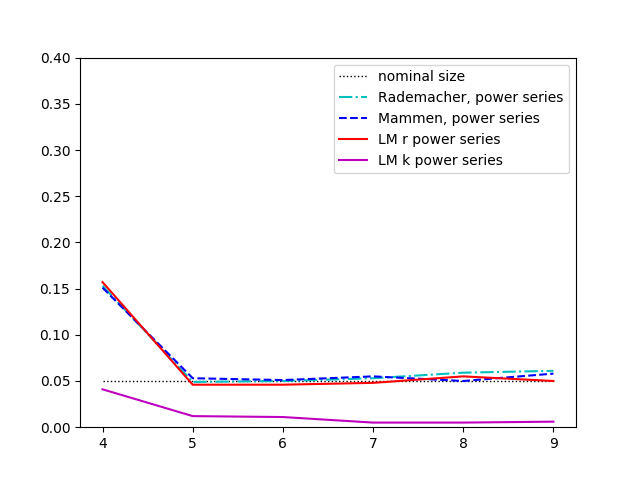} \includegraphics[scale=0.5]{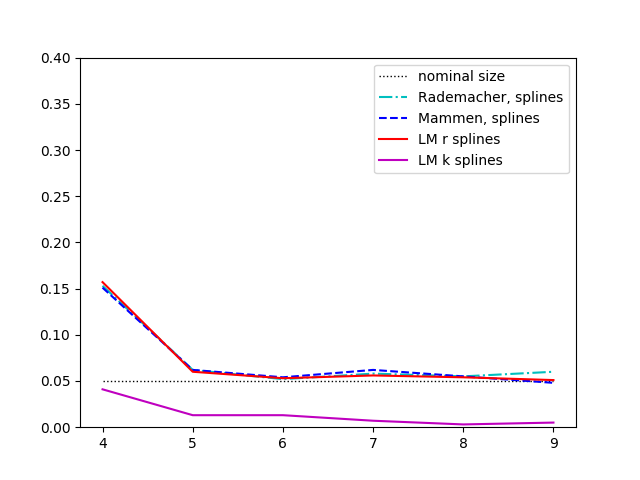}

\includegraphics[scale=0.5]{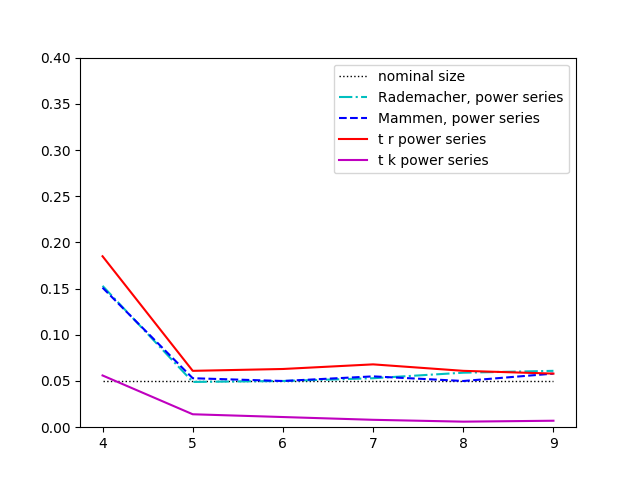} \includegraphics[scale=0.5]{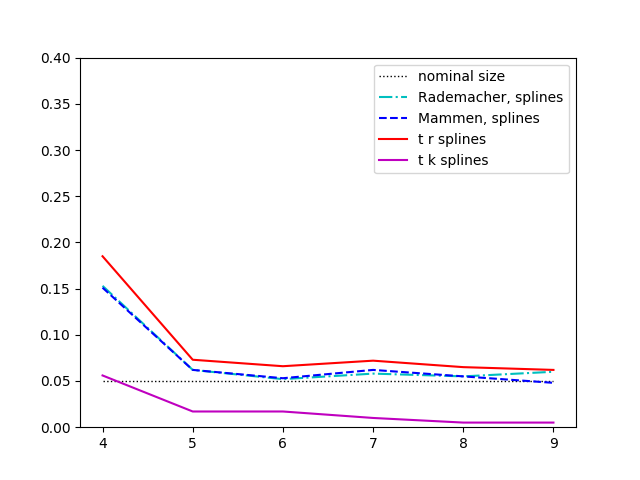}

\end{center}
\footnotesize{\vspace{-0.4cm}This figure plots the simulated size of the nominal 5\% test against the number of series terms in univariate series expansions, $a_n$ (including the constant). The left panel uses power series. The right panel uses splines. The upper panel uses the $\xi$ test statistic from Equation~\ref{xi}. The lower panel uses the $t$ test statistic from Equation~\ref{t_test_statistics}.

The red solid line corresponds to the test that uses the asymptotic critical values and normalization $\tau_n = r_n$. The magenta solid line corresponds to the test that uses the asymptotic critical values and normalization $\tau_n = k_n$. The cyan dash-dotted line corresponds to the test that uses the wild bootstrap critical values based on Rademacher distribution.  The blue dashed line corresponds to the test that uses the wild bootstrap critical values based on Mammen's distribution. The results are based on $M=1,000$ simulations and $B=399$ bootstrap iterations.}
\end{figure}

\begin{figure}[H]
\begin{center}
\caption{Simulated Power of the Test, $n=250$, $T = 2$}\label{fig_simulated_power_set1}
\includegraphics[scale=0.5]{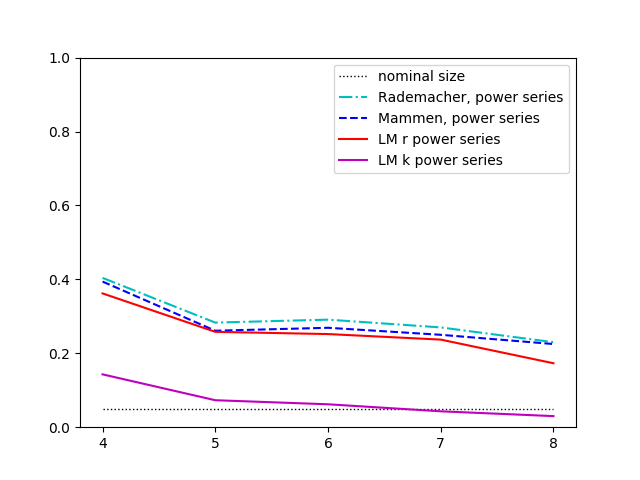} \includegraphics[scale=0.5]{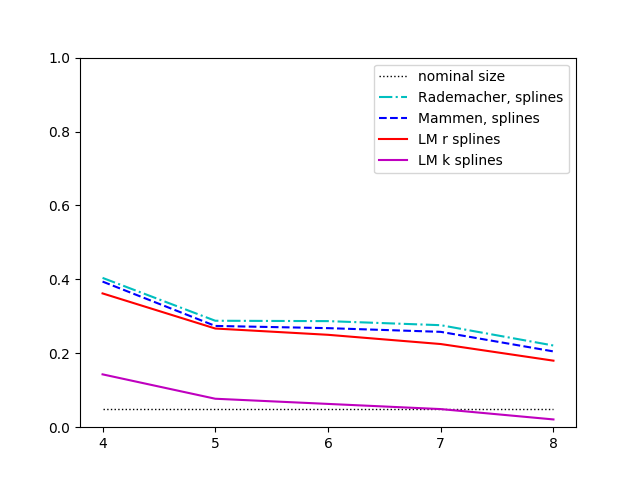}

\includegraphics[scale=0.5]{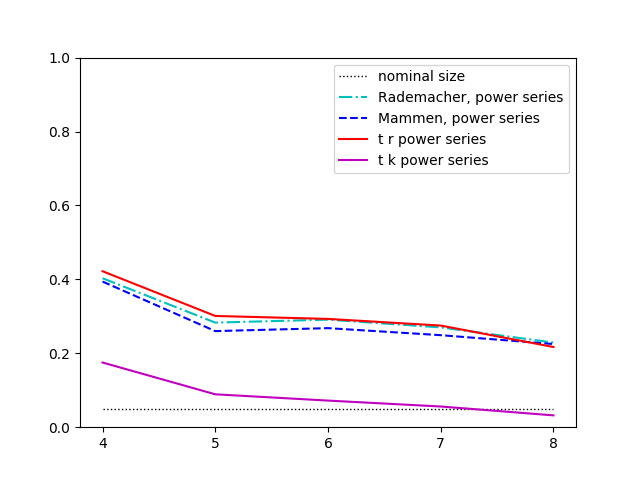} \includegraphics[scale=0.5]{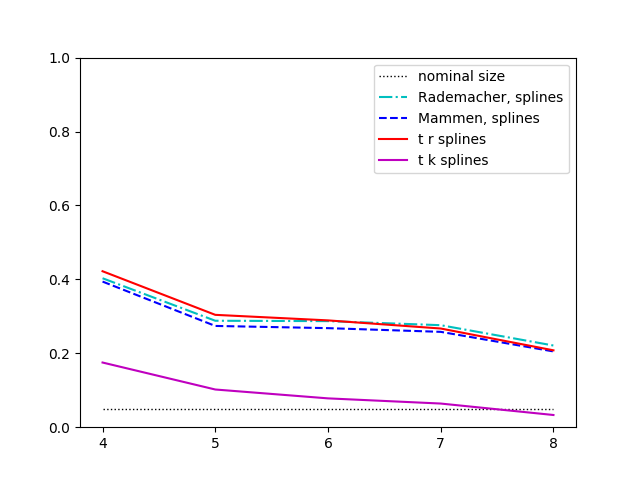}

\end{center}
\footnotesize{\vspace{-0.4cm}This figure plots the simulated power of the nominal 5\% test against the number of series terms in univariate series expansions, $a_n$ (including the constant). The left panel uses power series. The right panel uses splines. The upper panel uses the $\xi$ test statistic from Equation~\ref{xi}. The lower panel uses the $t$ test statistic from Equation~\ref{t_test_statistics}.

The red solid line corresponds to the test that uses the asymptotic critical values and normalization $\tau_n = r_n$. The magenta solid line corresponds to the test that uses the asymptotic critical values and normalization $\tau_n = k_n$. The cyan dash-dotted line corresponds to the test that uses the wild bootstrap critical values based on Rademacher distribution.  The blue dashed line corresponds to the test that uses the wild bootstrap critical values based on Mammen's distribution. The results are based on $M=1,000$ simulations and $B=399$ bootstrap iterations.}
\end{figure}

\begin{figure}[H]
\begin{center}
\caption{Simulated Power of the Test, $n=250$, $T = 4$}\label{fig_simulated_power_set2}
\includegraphics[scale=0.5]{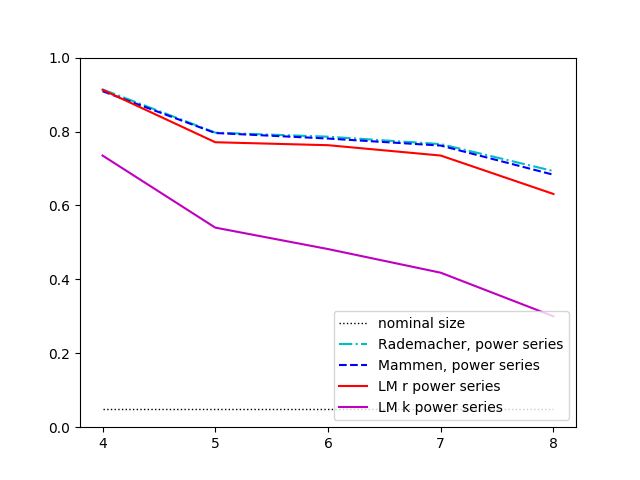} \includegraphics[scale=0.5]{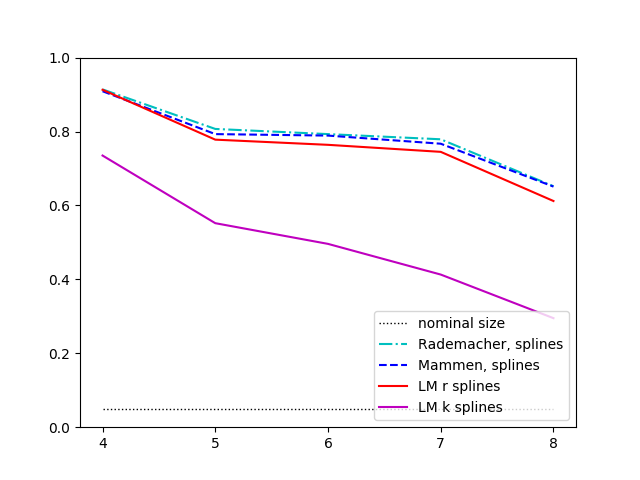}

\includegraphics[scale=0.5]{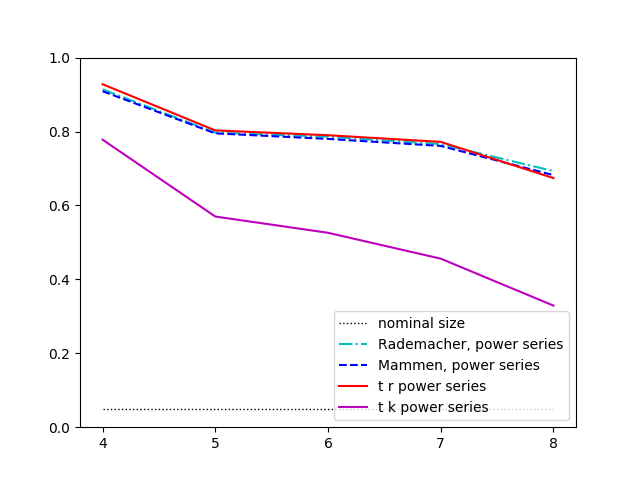} \includegraphics[scale=0.5]{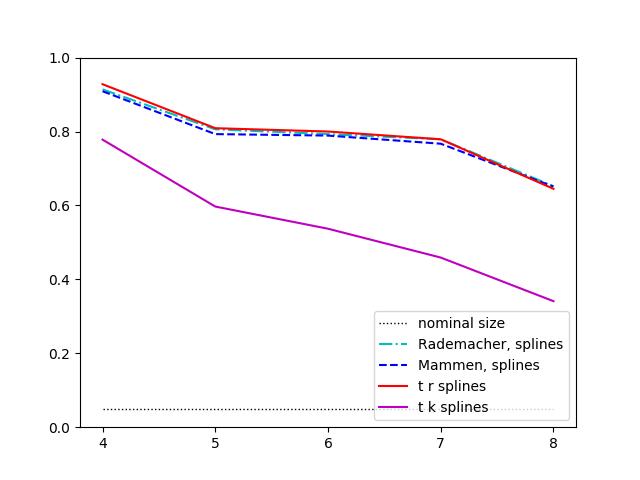}

\end{center}
\footnotesize{\vspace{-0.4cm}This figure plots the simulated power of the nominal 5\% test against the number of series terms in univariate series expansions, $a_n$ (including the constant). The left panel uses power series. The right panel uses splines. The upper panel uses the $\xi$ test statistic from Equation~\ref{xi}. The lower panel uses the $t$ test statistic from Equation~\ref{t_test_statistics}.

The red solid line corresponds to the test that uses the asymptotic critical values and normalization $\tau_n = r_n$. The magenta solid line corresponds to the test that uses the asymptotic critical values and normalization $\tau_n = k_n$. The cyan dash-dotted line corresponds to the test that uses the wild bootstrap critical values based on Rademacher distribution.  The blue dashed line corresponds to the test that uses the wild bootstrap critical values based on Mammen's distribution. The results are based on $M=1,000$ simulations and $B=399$ bootstrap iterations.}
\end{figure}

\begin{figure}[H]
\begin{center}
\caption{Simulated Power of the Test, $n=500$, $T = 2$}\label{fig_simulated_power_set3}
\includegraphics[scale=0.5]{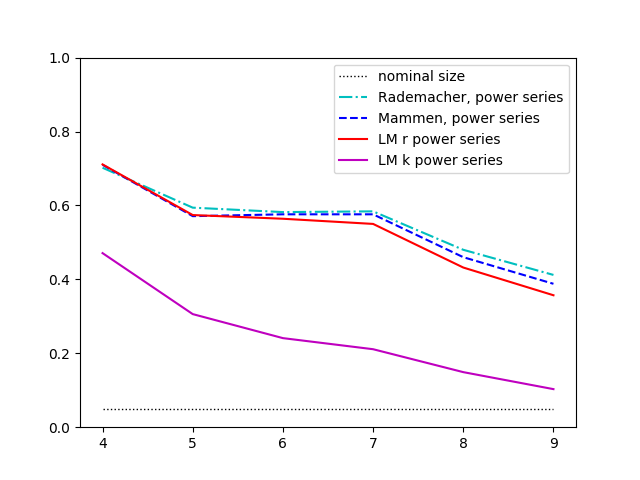} \includegraphics[scale=0.5]{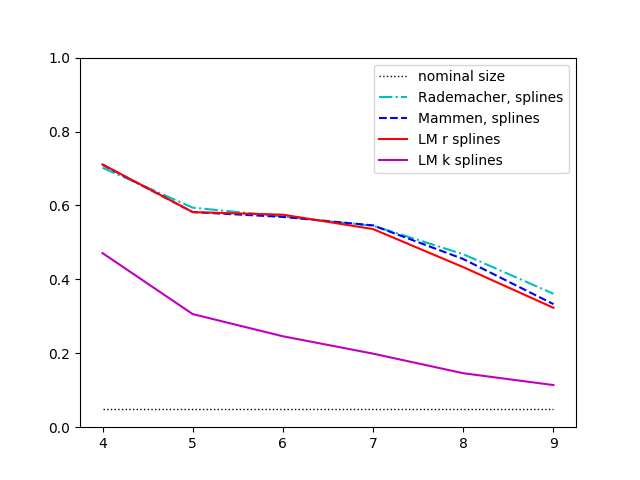}

\includegraphics[scale=0.5]{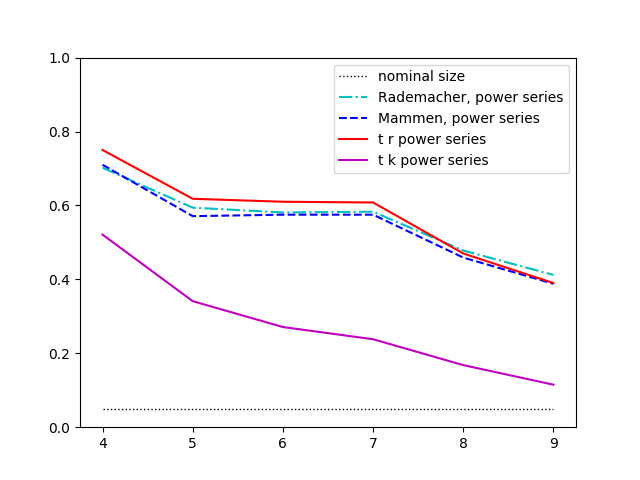} \includegraphics[scale=0.5]{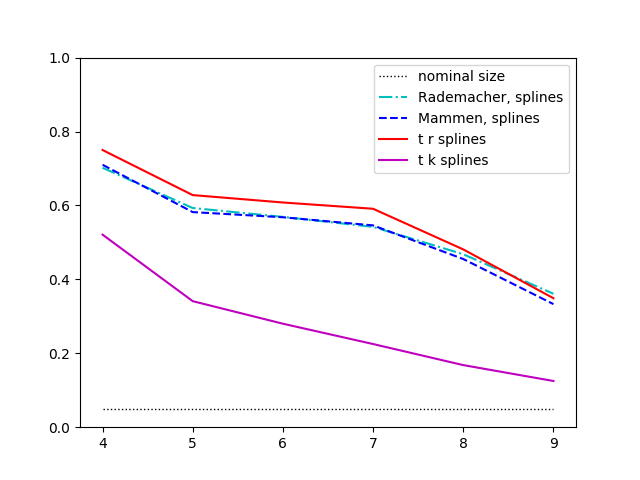}

\end{center}
\footnotesize{\vspace{-0.4cm}This figure plots the simulated power of the nominal 5\% test against the number of series terms in univariate series expansions, $a_n$ (including the constant). The left panel uses power series. The right panel uses splines. The upper panel uses the $\xi$ test statistic from Equation~\ref{xi}. The lower panel uses the $t$ test statistic from Equation~\ref{t_test_statistics}.

The red solid line corresponds to the test that uses the asymptotic critical values and normalization $\tau_n = r_n$. The magenta solid line corresponds to the test that uses the asymptotic critical values and normalization $\tau_n = k_n$. The cyan dash-dotted line corresponds to the test that uses the wild bootstrap critical values based on Rademacher distribution.  The blue dashed line corresponds to the test that uses the wild bootstrap critical values based on Mammen's distribution. The results are based on $M=1,000$ simulations and $B=399$ bootstrap iterations.}
\end{figure}

\begin{figure}[H]
\begin{center}
\caption{Simulated Power of the Test, $n=500$, $T = 4$}\label{fig_simulated_power_set4}
\includegraphics[scale=0.5]{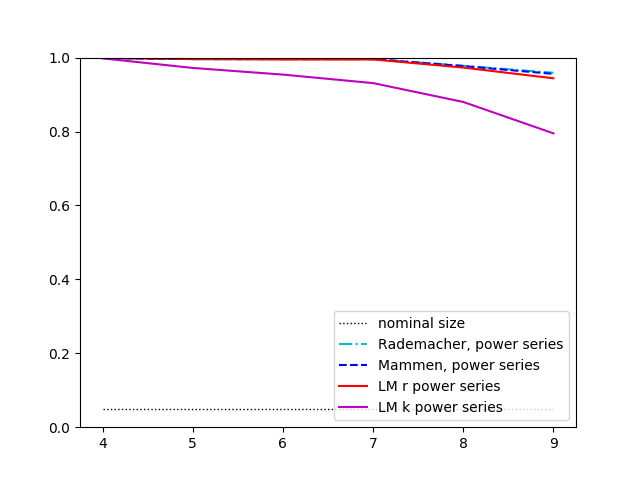} \includegraphics[scale=0.5]{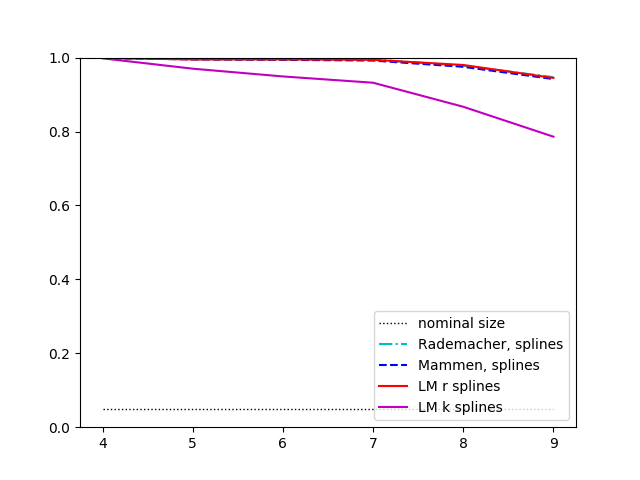}

\includegraphics[scale=0.5]{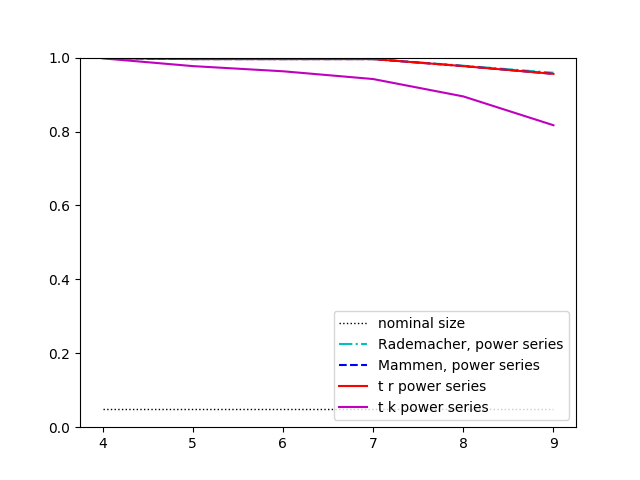} \includegraphics[scale=0.5]{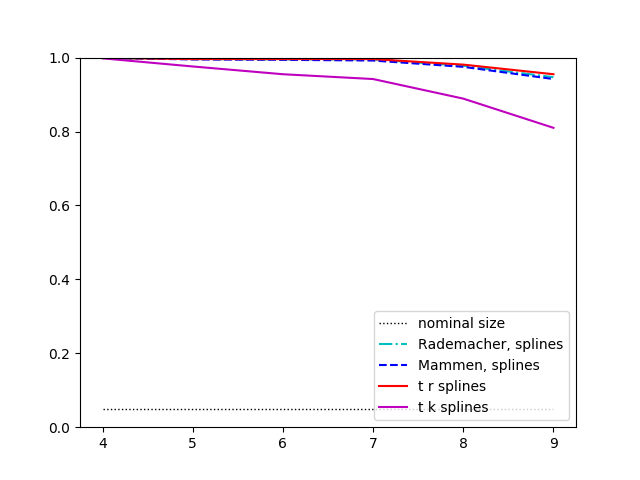}

\end{center}
\footnotesize{\vspace{-0.4cm}This figure plots the simulated power of the nominal 5\% test against the number of series terms in univariate series expansions, $a_n$ (including the constant). The left panel uses power series. The right panel uses splines. The upper panel uses the $\xi$ test statistic from Equation~\ref{xi}. The lower panel uses the $t$ test statistic from Equation~\ref{t_test_statistics}.

The red solid line corresponds to the test that uses the asymptotic critical values and normalization $\tau_n = r_n$. The magenta solid line corresponds to the test that uses the asymptotic critical values and normalization $\tau_n = k_n$. The cyan dash-dotted line corresponds to the test that uses the wild bootstrap critical values based on Rademacher distribution.  The blue dashed line corresponds to the test that uses the wild bootstrap critical values based on Mammen's distribution. The results are based on $M=1,000$ simulations and $B=399$ bootstrap iterations.}
\end{figure}

\begin{figure}[H]
\begin{center}
\caption{Simulated Size of the Test, Heteroskedastic Errors, $n=250$, $T = 2$}\label{fig_simulated_size_set1_hc}
\includegraphics[scale=0.5]{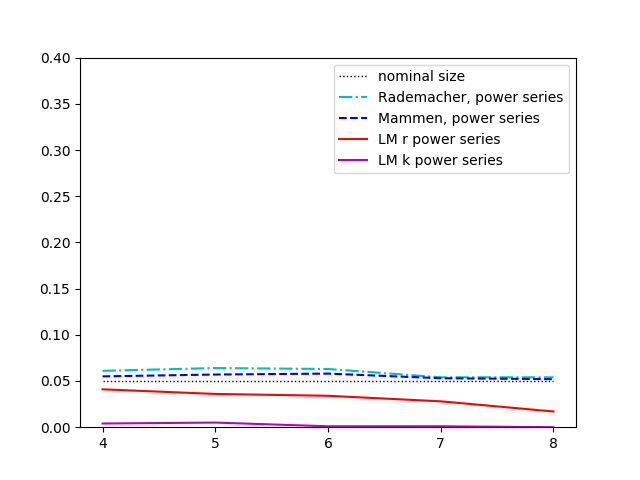} \includegraphics[scale=0.5]{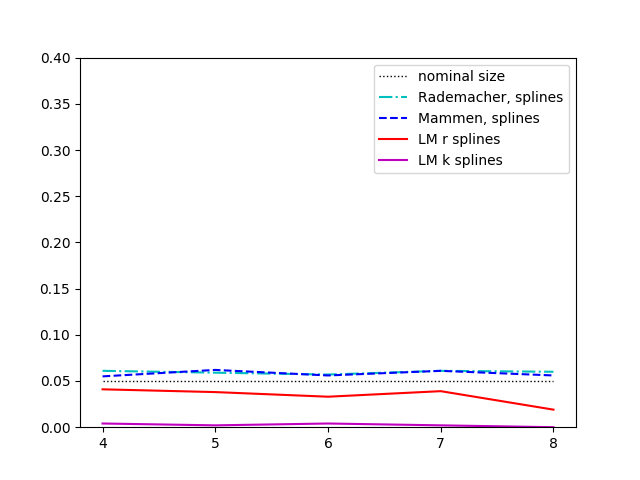}

\includegraphics[scale=0.5]{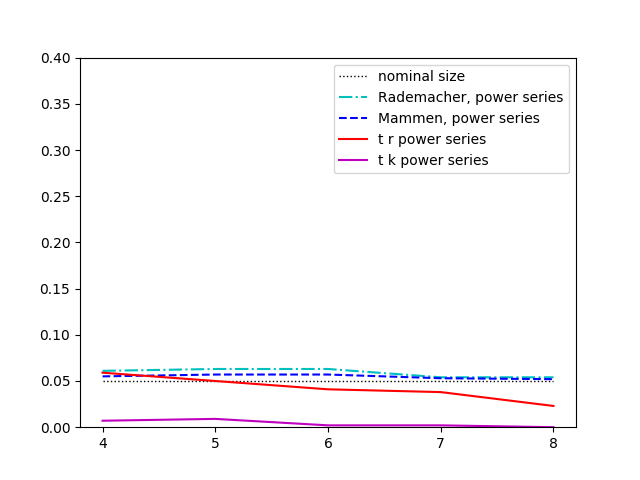} \includegraphics[scale=0.5]{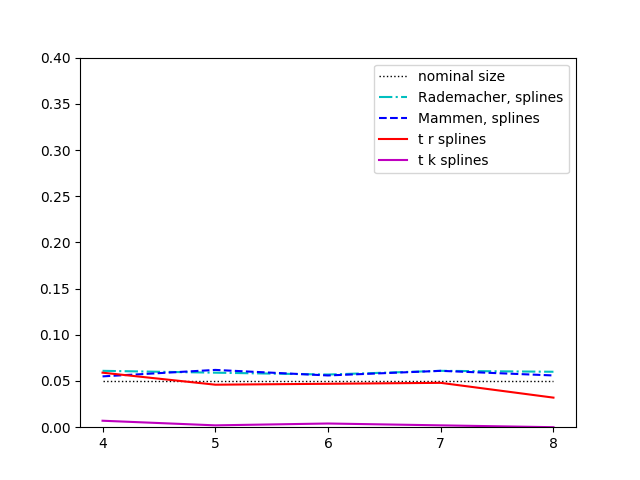}

\end{center}
\footnotesize{\vspace{-0.4cm}This figure plots the simulated size of the nominal 5\% test against the number of series terms in univariate series expansions, $a_n$ (including the constant). The left panel uses power series. The right panel uses splines. The upper panel uses the $\xi_{HC}$ test statistic from Equation~\ref{xi_hc}. The lower panel uses the $t$ test statistic from Equation~\ref{t_test_statistics}.

The red solid line corresponds to the test that uses the asymptotic critical values and normalization $\tau_n = r_n$. The magenta solid line corresponds to the test that uses the asymptotic critical values and normalization $\tau_n = k_n$. The cyan dash-dotted line corresponds to the test that uses the wild bootstrap critical values based on Rademacher distribution.  The blue dashed line corresponds to the test that uses the wild bootstrap critical values based on Mammen's distribution. The results are based on $M=1,000$ simulations and $B=399$ bootstrap iterations.}
\end{figure}

\begin{figure}[H]
\begin{center}
\caption{Simulated Size of the Test, Heteroskedastic Errors, $n=250$, $T = 4$}\label{fig_simulated_size_set2_hc}
\includegraphics[scale=0.5]{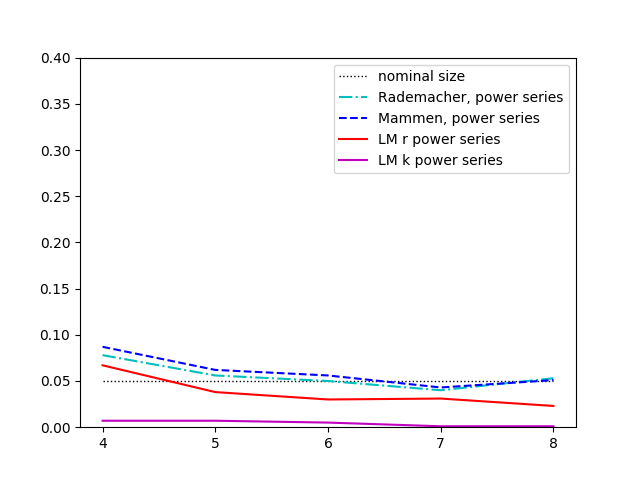} \includegraphics[scale=0.5]{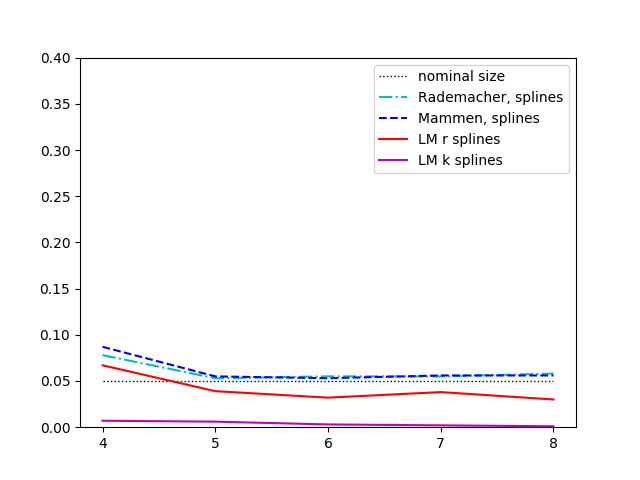}

\includegraphics[scale=0.5]{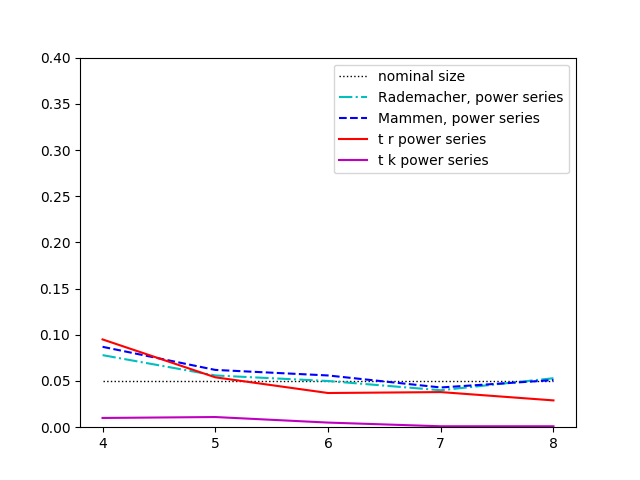} \includegraphics[scale=0.5]{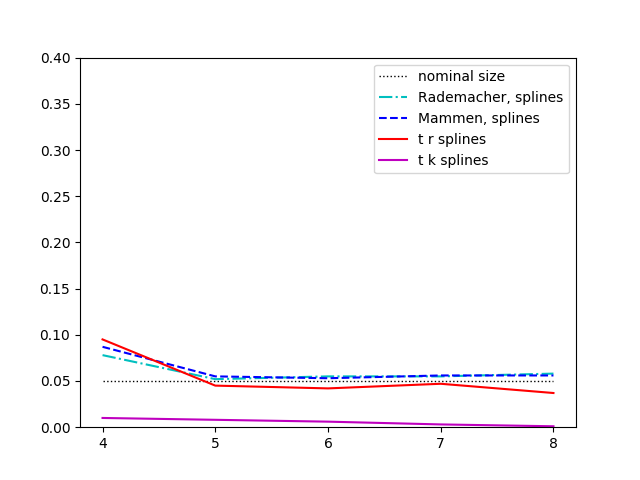}

\end{center}
\footnotesize{\vspace{-0.4cm}This figure plots the simulated size of the nominal 5\% test against the number of series terms in univariate series expansions, $a_n$ (including the constant). The left panel uses power series. The right panel uses splines. The upper panel uses the $\xi_{HC}$ test statistic from Equation~\ref{xi_hc}. The lower panel uses the $t$ test statistic from Equation~\ref{t_test_statistics}.

The red solid line corresponds to the test that uses the asymptotic critical values and normalization $\tau_n = r_n$. The magenta solid line corresponds to the test that uses the asymptotic critical values and normalization $\tau_n = k_n$. The cyan dash-dotted line corresponds to the test that uses the wild bootstrap critical values based on Rademacher distribution.  The blue dashed line corresponds to the test that uses the wild bootstrap critical values based on Mammen's distribution. The results are based on $M=1,000$ simulations and $B=399$ bootstrap iterations.}
\end{figure}

\begin{figure}[H]
\begin{center}
\caption{Simulated Size of the Test, Heteroskedastic Errors, $n=500$, $T = 2$}\label{fig_simulated_size_set3_hc}
\includegraphics[scale=0.5]{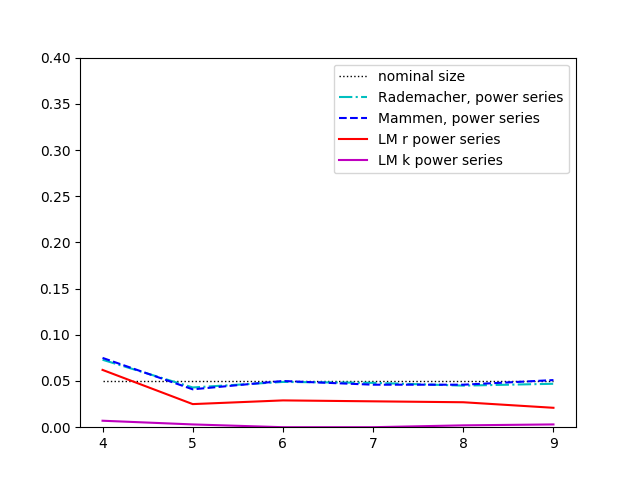} \includegraphics[scale=0.5]{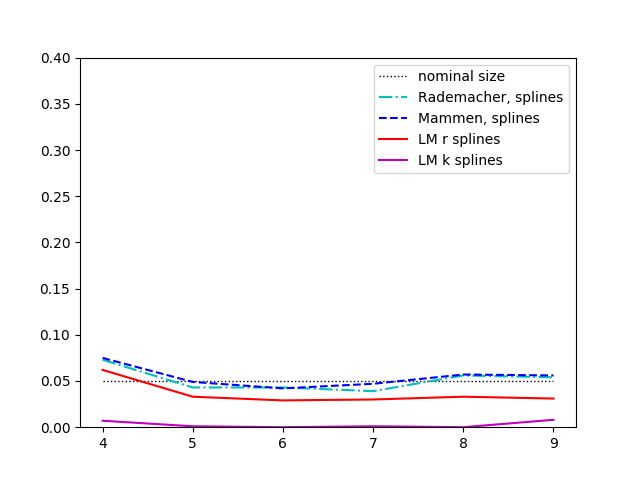}

\includegraphics[scale=0.5]{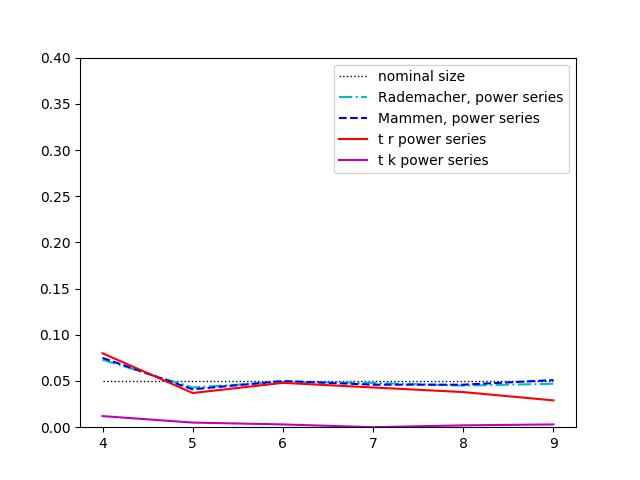} \includegraphics[scale=0.5]{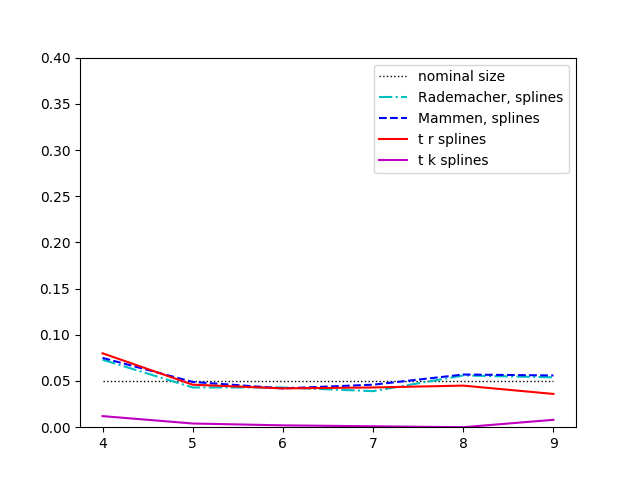}

\end{center}
\footnotesize{\vspace{-0.4cm}This figure plots the simulated size of the nominal 5\% test against the number of series terms in univariate series expansions, $a_n$ (including the constant). The left panel uses power series. The right panel uses splines. The upper panel uses the $\xi_{HC}$ test statistic from Equation~\ref{xi_hc}. The lower panel uses the $t$ test statistic from Equation~\ref{t_test_statistics}.

The red solid line corresponds to the test that uses the asymptotic critical values and normalization $\tau_n = r_n$. The magenta solid line corresponds to the test that uses the asymptotic critical values and normalization $\tau_n = k_n$. The cyan dash-dotted line corresponds to the test that uses the wild bootstrap critical values based on Rademacher distribution.  The blue dashed line corresponds to the test that uses the wild bootstrap critical values based on Mammen's distribution. The results are based on $M=1,000$ simulations and $B=399$ bootstrap iterations.}
\end{figure}

\begin{figure}[H]
\begin{center}
\caption{Simulated Size of the Test, Heteroskedastic Errors, $n=500$, $T = 4$}\label{fig_simulated_size_set4_hc}
\includegraphics[scale=0.5]{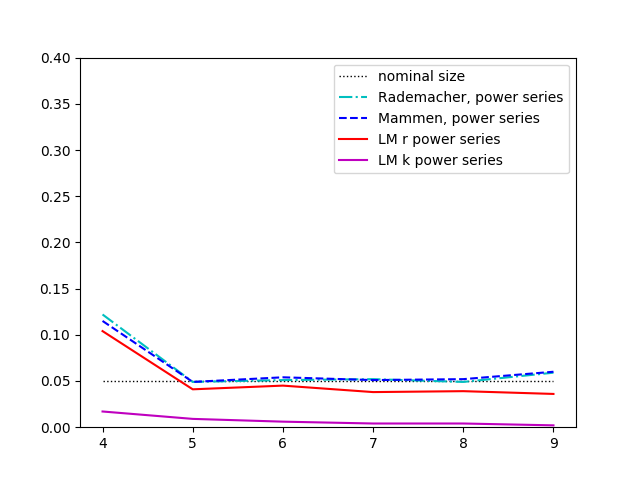} \includegraphics[scale=0.5]{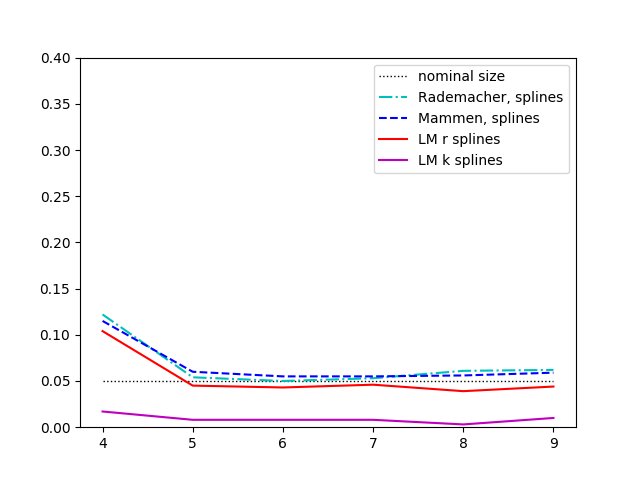}

\includegraphics[scale=0.5]{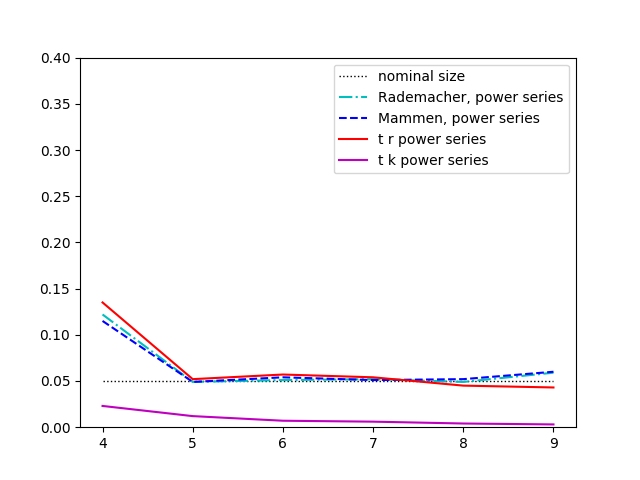} \includegraphics[scale=0.5]{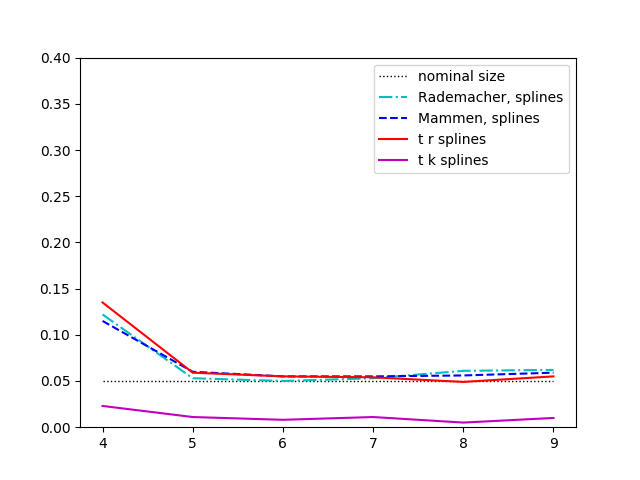}

\end{center}
\footnotesize{\vspace{-0.4cm}This figure plots the simulated size of the nominal 5\% test against the number of series terms in univariate series expansions, $a_n$ (including the constant). The left panel uses power series. The right panel uses splines. The upper panel uses the $\xi_{HC}$ test statistic from Equation~\ref{xi_hc}. The lower panel uses the $t$ test statistic from Equation~\ref{t_test_statistics}.

The red solid line corresponds to the test that uses the asymptotic critical values and normalization $\tau_n = r_n$. The magenta solid line corresponds to the test that uses the asymptotic critical values and normalization $\tau_n = k_n$. The cyan dash-dotted line corresponds to the test that uses the wild bootstrap critical values based on Rademacher distribution.  The blue dashed line corresponds to the test that uses the wild bootstrap critical values based on Mammen's distribution. The results are based on $M=1,000$ simulations and $B=399$ bootstrap iterations.}
\end{figure}

\begin{figure}[H]
\begin{center}
\caption{Simulated Power of the Test, Heteroskedastic Errors, $n=250$, $T = 2$}\label{fig_simulated_power_set1_hc}
\includegraphics[scale=0.5]{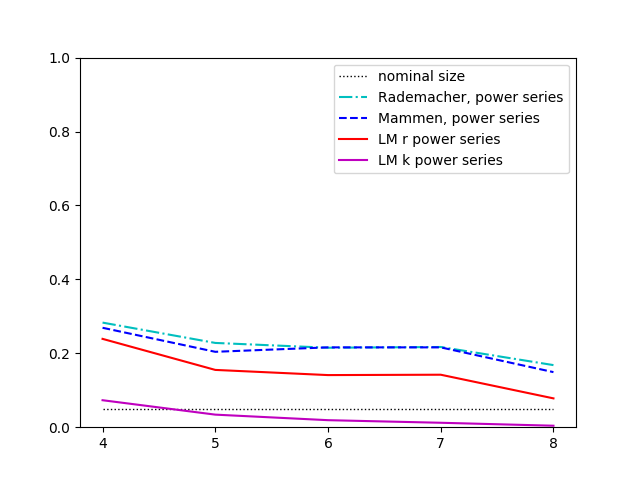} \includegraphics[scale=0.5]{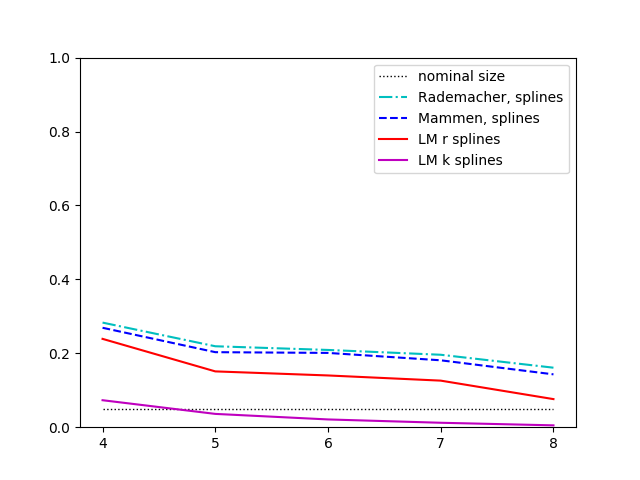}

\includegraphics[scale=0.5]{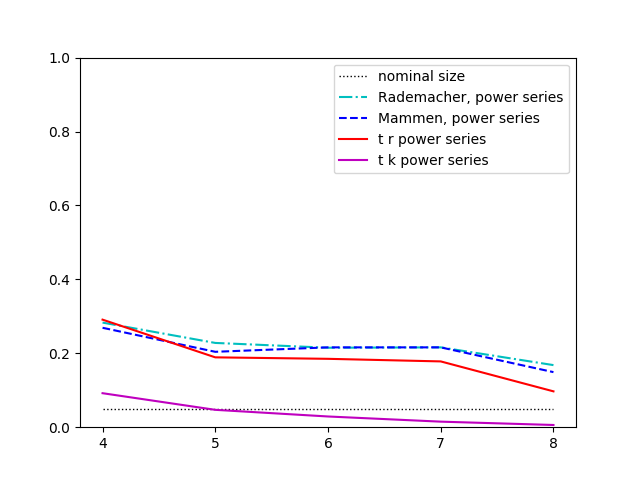} \includegraphics[scale=0.5]{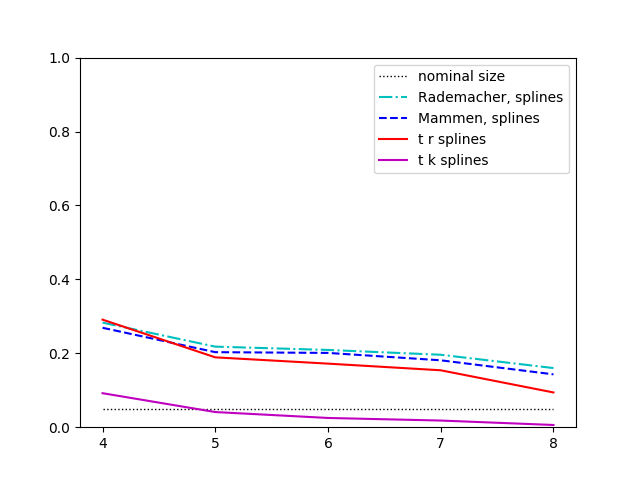}

\end{center}
\footnotesize{\vspace{-0.4cm}This figure plots the simulated power of the nominal 5\% test against the number of series terms in univariate series expansions, $a_n$ (including the constant). The left panel uses power series. The right panel uses splines. The upper panel uses the $\xi_{HC}$ test statistic from Equation~\ref{xi_hc}. The lower panel uses the $t$ test statistic from Equation~\ref{t_test_statistics}.

The red solid line corresponds to the test that uses the asymptotic critical values and normalization $\tau_n = r_n$. The magenta solid line corresponds to the test that uses the asymptotic critical values and normalization $\tau_n = k_n$. The cyan dash-dotted line corresponds to the test that uses the wild bootstrap critical values based on Rademacher distribution.  The blue dashed line corresponds to the test that uses the wild bootstrap critical values based on Mammen's distribution. The results are based on $M=1,000$ simulations and $B=399$ bootstrap iterations.}
\end{figure}

\begin{figure}[H]
\begin{center}
\caption{Simulated Power of the Test, Heteroskedastic Errors, $n=250$, $T = 4$}\label{fig_simulated_power_set2_hc}
\includegraphics[scale=0.5]{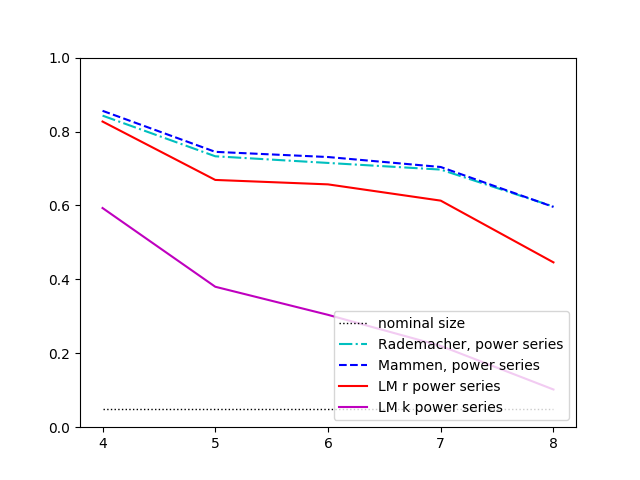} \includegraphics[scale=0.5]{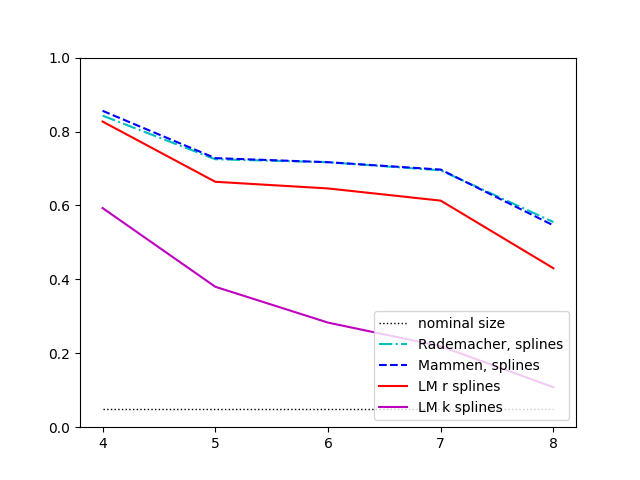}

\includegraphics[scale=0.5]{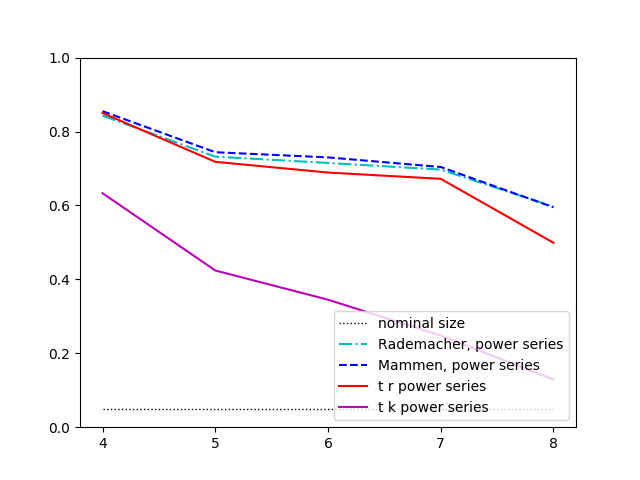} \includegraphics[scale=0.5]{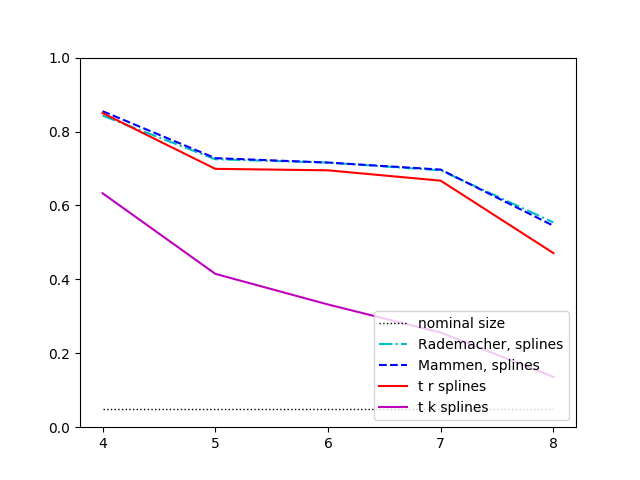}

\end{center}
\footnotesize{\vspace{-0.4cm}This figure plots the simulated power of the nominal 5\% test against the number of series terms in univariate series expansions, $a_n$ (including the constant). The left panel uses power series. The right panel uses splines. The upper panel uses the $\xi_{HC}$ test statistic from Equation~\ref{xi_hc}. The lower panel uses the $t$ test statistic from Equation~\ref{t_test_statistics}.

The red solid line corresponds to the test that uses the asymptotic critical values and normalization $\tau_n = r_n$. The magenta solid line corresponds to the test that uses the asymptotic critical values and normalization $\tau_n = k_n$. The cyan dash-dotted line corresponds to the test that uses the wild bootstrap critical values based on Rademacher distribution.  The blue dashed line corresponds to the test that uses the wild bootstrap critical values based on Mammen's distribution. The results are based on $M=1,000$ simulations and $B=399$ bootstrap iterations.}
\end{figure}

\begin{figure}[H]
\begin{center}
\caption{Simulated Power of the Test, Heteroskedastic Errors, $n=500$, $T = 2$}\label{fig_simulated_power_set3_hc}
\includegraphics[scale=0.5]{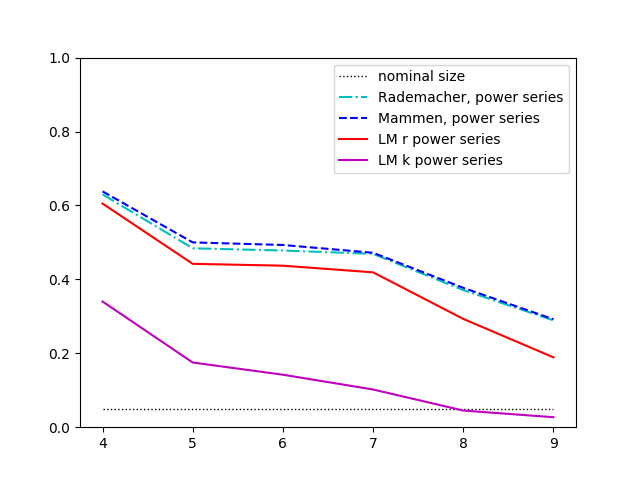} \includegraphics[scale=0.5]{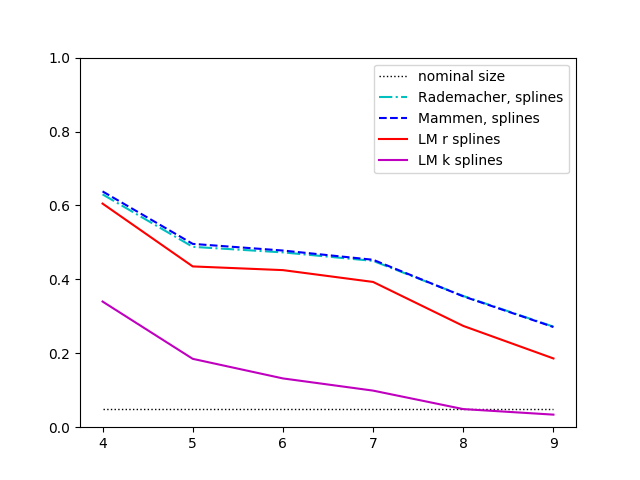}

\includegraphics[scale=0.5]{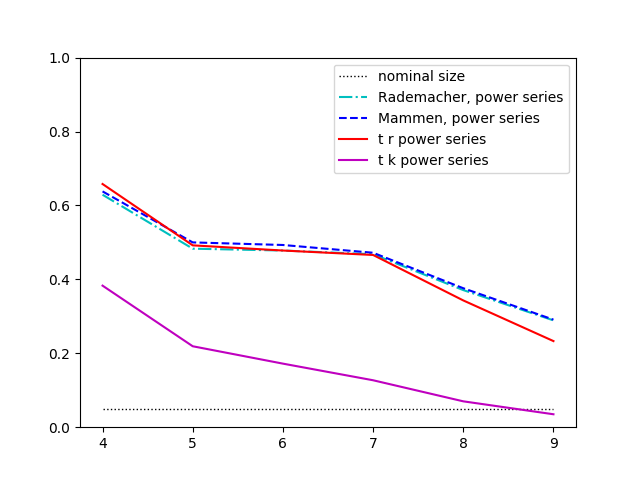} \includegraphics[scale=0.5]{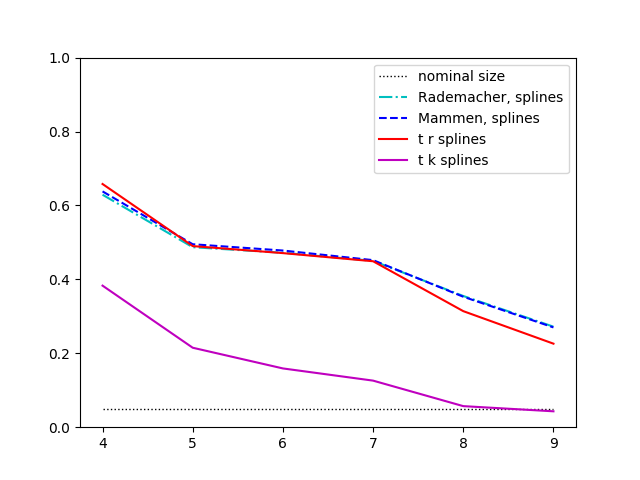}

\end{center}
\footnotesize{\vspace{-0.4cm}This figure plots the simulated power of the nominal 5\% test against the number of series terms in univariate series expansions, $a_n$ (including the constant). The left panel uses power series. The right panel uses splines. The upper panel uses the $\xi_{HC}$ test statistic from Equation~\ref{xi_hc}. The lower panel uses the $t$ test statistic from Equation~\ref{t_test_statistics}.

The red solid line corresponds to the test that uses the asymptotic critical values and normalization $\tau_n = r_n$. The magenta solid line corresponds to the test that uses the asymptotic critical values and normalization $\tau_n = k_n$. The cyan dash-dotted line corresponds to the test that uses the wild bootstrap critical values based on Rademacher distribution.  The blue dashed line corresponds to the test that uses the wild bootstrap critical values based on Mammen's distribution. The results are based on $M=1,000$ simulations and $B=399$ bootstrap iterations.}
\end{figure}

\begin{figure}[H]
\begin{center}
\caption{Simulated Power of the Test, Heteroskedastic Errors, $n=500$, $T = 4$}\label{fig_simulated_power_set4_hc}
\includegraphics[scale=0.5]{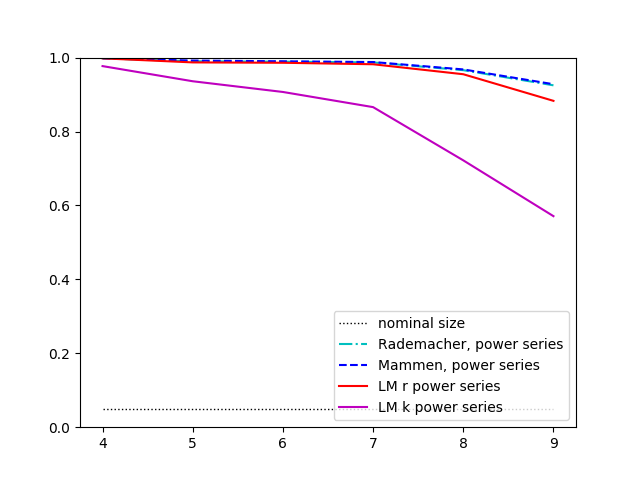} \includegraphics[scale=0.5]{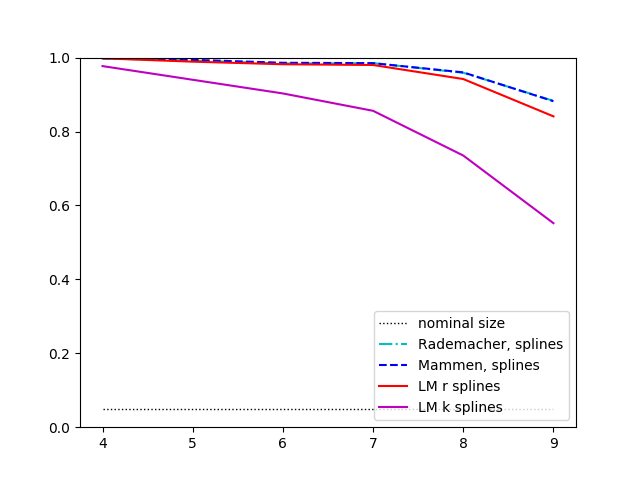}

\includegraphics[scale=0.5]{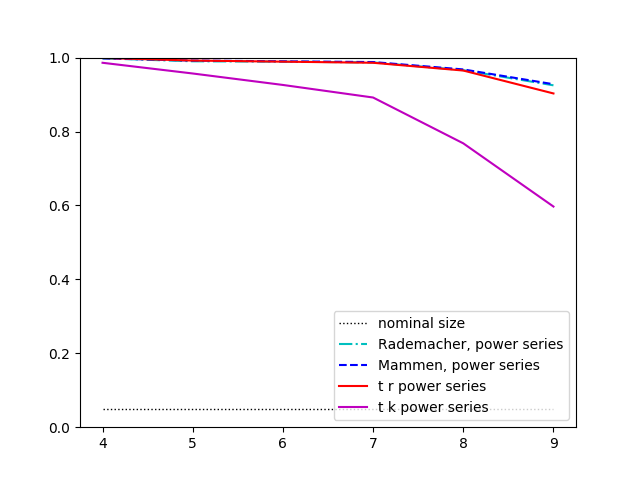} \includegraphics[scale=0.5]{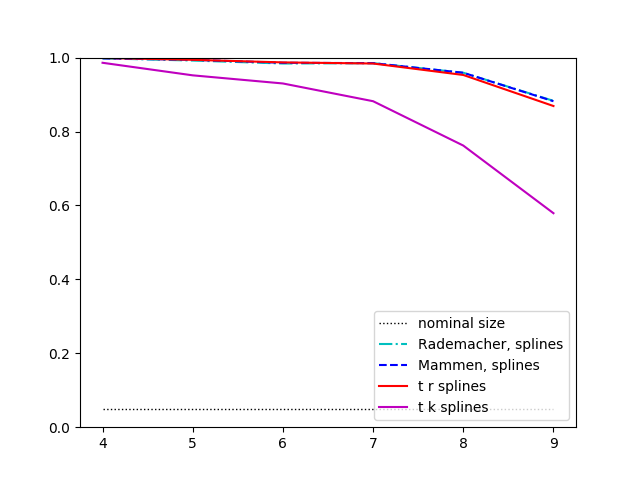}

\end{center}
\footnotesize{\vspace{-0.4cm}This figure plots the simulated power of the nominal 5\% test against the number of series terms in univariate series expansions, $a_n$ (including the constant). The left panel uses power series. The right panel uses splines. The upper panel uses the $\xi_{HC}$ test statistic from Equation~\ref{xi_hc}. The lower panel uses the $t$ test statistic from Equation~\ref{t_test_statistics}.

The red solid line corresponds to the test that uses the asymptotic critical values and normalization $\tau_n = r_n$. The magenta solid line corresponds to the test that uses the asymptotic critical values and normalization $\tau_n = k_n$. The cyan dash-dotted line corresponds to the test that uses the wild bootstrap critical values based on Rademacher distribution.  The blue dashed line corresponds to the test that uses the wild bootstrap critical values based on Mammen's distribution. The results are based on $M=1,000$ simulations and $B=399$ bootstrap iterations.}
\end{figure}

\begin{figure}[H]
\begin{center}
\caption{Experience Effects from Different Methods}\label{figure_exp_effects}
\includegraphics[scale=0.8]{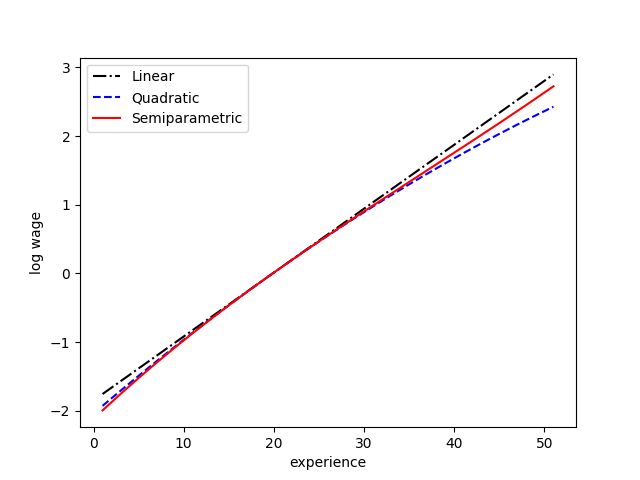}
\end{center}

\begin{singlespace}\footnotesize{This figure plots the estimated effects of experience from different models. The dash-dotted line shows the estimated experience effect from the linear model. The dashed line shows the estimated experience effect from the quadratic model. The solid line shows the estimated experience effect from the semiparametric model. The $x$ axis plots experience. The $y$ axis plots the logarithm of wage.}
\end{singlespace}
\end{figure}

\begin{table}[h]
\begin{center}
\caption{Simulated Size and Power of Data-Driven Test}\label{Tbl_data_driven}
\begin{tabular}{l | c  c  c | c c c}
 & \multicolumn{3}{c}{Power Series} & \multicolumn{3}{c}{Splines} \\ 
	&	Data-Driven	&	$a_n = 4$	&	$a_n = 9$	&	Data-Driven	&	$a_n = 4$	&	$a_n = 9$	\\ \hline\hline
	& \multicolumn{6}{c}{$n = 250, T = 2$} \\
Size	&	0.055	&	0.046	&	0.049	&	0.054	&	0.046	&	0.047	\\
Power, regular	&	0.369	&	0.363	&	0.196	&	0.369	&	0.363	&	0.201	\\
Power, orthogonal	&	0.105	&	0.051	&	0.224	&	0.096	&	0.051	&	0.225	\\ \hline\hline
													
	& \multicolumn{6}{c}{$n = 250, T = 4$} \\
Size	&	0.064	&	0.061	&	0.041	&	0.063	&	0.061	&	0.044	\\
Power, regular	&	0.825	&	0.825	&	0.622	&	0.825	&	0.825	&	0.621	\\
Power, orthogonal	&	0.447	&	0.062	&	0.654	&	0.440	&	0.062	&	0.644	\\\hline\hline
													
	& \multicolumn{6}{c}{$n = 500, T = 2$} \\
Size	&	0.058	&	0.054	&	0.055	&	0.056	&	0.054	&	0.054	\\
Power, regular	&	0.673	&	0.671	&	0.444	&	0.673	&	0.671	&	0.448	\\
Power, orthogonal	&	0.238	&	0.048	&	0.421	&	0.241	&	0.048	&	0.426	\\ \hline\hline
													
	& \multicolumn{6}{c}{$n = 500, T = 4$} \\
Size	&	0.041	&	0.037	&	0.052	&	0.041	&	0.037	&	0.051	\\
Power, regular &	0.990	&	0.990	&	0.935	&	0.990	&	0.990	&	0.932	\\
Power, orthogonal	&	0.856	&	0.033	&	0.933	&	0.843	&	0.033	&	0.939
\end{tabular}
\end{center}
\begin{singlespace}\footnotesize{The table reports the simulated size, power against the regular alternative, and power against the orthogonal alternative of the test based on the test statistic $\xi$. The left panel uses power series and the right panel uses splines. The column called ``Data-Driven'' uses the data-driven value of the tuning parameter $r_n$ as described in Section~\ref{simulations_data_driven}. Other columns use the fixed number of series terms $a_n$ equal to 4 or 9 (including the constant). The results are based on $M=1,000$ simulations.}
\end{singlespace}
\end{table}

\begin{table}[h]
\begin{center}
\caption{Specification Testing Results}\label{Tbl_testing}
\begin{tabular}{l | c | c | c}
& Test Statistic & 5\% Critical Value & 10\% Critical Value \\ \hline\hline
 \multicolumn{4}{c}{Quadratic Model} \\
$\xi_{HC}$ & 19.835 & 21.026 & 18.549 \\
$t_{HC}$ & 1.599 & 1.645 & 1.282 \\ \hline\hline
 \multicolumn{4}{c}{Linear Model} \\
$\xi_{HC}$ & 42.318 & 22.362 & 19.812 \\
$t_{HC}$ & 5.750 & 1.645 & 1.282 \\ \hline\hline
 \multicolumn{4}{c}{Semiparametric Model} \\
$\xi_{HC}$ & 15.957 & 19.675 & 17.275 \\
$t_{HC}$ & 1.057 & 1.645 & 1.282
\end{tabular}
\end{center}
\begin{singlespace}\footnotesize{The table reports the values of the test statistics $\xi_{HC} \overset{a}{\sim} \chi^2(r_n)$ and $t_{HC} \overset{a}{\sim} N(0,1)$ for the quadratic specification from Equation~\ref{quadratic_model}, linear specification, and semiparametric specification from Equation~\ref{semiparametric_model}. The number of restrictions is $r_n = 12$, $r_n = 13$, and $r_n = 11$ respectively. The corresponding critical values are shown together with the test statistics.}
\end{singlespace}
\end{table}

\clearpage

\section{Proofs}\label{appendix_proofs}

\subsection{Proof of Theorem~\ref{asy_distr_t_r_n_hc}}

The homoskedastic and heteroskedastic test statistics have different expressions, and it is convenient to introduce some notation that allows me to write them in a similar form. Denote $\hat{\Omega} = n^{-1} \sum_{i=1}^{n}{\tilde{Z}_i' \tilde{e}_i \tilde{e}_i' \tilde{Z}_i}$ in the heteroskedastic case or $\hat{\Omega} = n^{-1} \sum_{i=1}^{n}{\tilde{Z}_i' \tilde{\Sigma}_T \tilde{Z}_i}$ in the homoskedastic case. Then both test statistics can be written as
\[
\xi = \left( \sum_{i=1}^{n}{\tilde{e}_i' \tilde{Z}_i} \right) \left( n \hat{\Omega} \right)^{-1} \left( \sum_{i=1}^{n}{\tilde{Z}_i' \tilde{e}_i} \right)  = \tilde{e}' \tilde{Z}  \left( n \hat{\Omega} \right)^{-1} \tilde{Z}' \tilde{e} = (\hat{\varepsilon} + \hat{R})' M_{\hat{W}} \tilde{Z}  \left( n \hat{\Omega} \right)^{-1} \tilde{Z}' M_{\hat{W}} (\hat{\varepsilon} + \hat{R})
\]

Because $\tilde{Z} = M_{\hat{W}} \hat{Z}$ and $M_{\hat{W}} \tilde{Z} = M_{\hat{W}} M_{\hat{W}} \hat{Z} = M_{\hat{W}} \hat{Z} = \tilde{Z}$, the test statistic can be rewritten as
\[
\xi  = (\hat{\varepsilon} + \hat{R})' \tilde{Z}  \left( n \hat{\Omega} \right)^{-1} \tilde{Z}' (\hat{\varepsilon} + \hat{R})
\]

The proof consists of several steps.

Step 1. Decompose the test statistic and bound the remainder terms.
\[
(\hat{\varepsilon} + \hat{R})' \tilde{Z}  \left( n \hat{\Omega} \right)^{-1} \tilde{Z}' (\hat{\varepsilon} + \hat{R}) = 
 \hat{\varepsilon}' \tilde{Z}  \left( n \hat{\Omega} \right)^{-1} \tilde{Z}' \hat{\varepsilon} +  \hat{R}'  \tilde{Z}  \left( n \hat{\Omega} \right)^{-1} \tilde{Z}'  \hat{R}
 + 2  \hat{R}' \tilde{Z}  \left( n \hat{\Omega} \right)^{-1} \tilde{Z}' \hat{\varepsilon}
\]

By Lemma~\ref{lemma_PP}, the smallest and largest eigenvalues of $\tilde{Z} \tilde{Z}'/(nT)$ converge to one. Because $\tilde{Z}' \tilde{Z}/(nT)$ and $\tilde{Z} \tilde{Z}'/(nT)$ have the same nonzero eigenvalues, $\lambda_{\max}(\tilde{Z} \tilde{Z}'/(nT))$ converges in probability to 1. Moreover, the eigenvalues of $\hat{\Omega}$ are bounded below and above. Thus, by Assumption~\ref{series_approx},
\[
\hat{R}' \tilde{Z} (n \hat{\Omega})^{-1} \tilde{Z}' \hat{R} \leq T C \hat{R}' ((nT)^{-1} \tilde{Z} \tilde{Z}') \hat{R} \leq C \hat{R}' \hat{R} = O_p(n m_n^{-2\alpha}),
\]
where $T$ is absorbed by $C$ because the length of the panel is fixed.

Next,
\[
\Big{|} \hat{R}' \tilde{Z} (n \hat{\Omega})^{-1} \tilde{Z}' \hat{\varepsilon} \Big{|} \leq \Big{|} T C \lambda_{\max}(\tilde{Z} \tilde{Z}'/(nT)) \hat{R}' \hat{\varepsilon} \Big{|} \leq \Big{|} C \hat{R}' \hat{\varepsilon} \Big{|} = O_p(n^{1/2} m_n^{-\alpha})
\]

Thus,
\[
(\hat{\varepsilon} + \hat{R})' \tilde{Z}  \left( n \hat{\Omega} \right)^{-1} \tilde{Z}' (\hat{\varepsilon} + \hat{R}) =  \hat{\varepsilon}' \tilde{Z}  \left( n \hat{\Omega} \right)^{-1} \tilde{Z}' \hat{\varepsilon} + O_p(n m_n^{-2\alpha}) + O_p(n^{1/2} m_n^{-\alpha})
\]

Step 2. Further decompose the leading term and bound the new remainder terms.
\[
\hat{\varepsilon}' \tilde{Z}  \left( n \hat{\Omega} \right)^{-1} \tilde{Z}' \hat{\varepsilon} = \hat{\varepsilon}' \hat{Z} (n \hat{\Omega})^{-1} \hat{Z}' \hat{\varepsilon} - 2 \hat{\varepsilon}' P_{\hat{W}} \hat{Z} (n \hat{\Omega})^{-1} \hat{Z}' \hat{\varepsilon} + \hat{\varepsilon}' P_{\hat{W}} \hat{Z} (n \hat{\Omega})^{-1} \hat{Z}' P_{\hat{W}} \hat{\varepsilon}
\]

Let $\chi_i = \hat{Z}_i' \hat{\varepsilon}_i$, $\Omega = E[\hat{Z}_i' \hat{\varepsilon}_i \hat{\varepsilon}_i' \hat{Z}_i] = E[\chi_i \chi_i']$. Note that
\begin{align*}
&E[(n^{-1} \hat{\varepsilon}' \hat{Z}) \Omega^{-1} (n^{-1} \hat{Z}' \hat{\varepsilon})]/n =E[\chi_i' \Omega^{-1} \chi_i] = E[tr(\chi_i' \Omega^{-1} \chi_i)]/n \\
&= E[tr(\Omega^{-1} \chi_i \chi_i')]/n = tr(\Omega^{-1} E[\chi_i \chi_i'])/n = tr(I_{r_n}) /n= r_n/n
\end{align*}

Thus, by Markov's inequality,
\[
\| \Omega^{-1} (n^{-1} \hat{Z}' \hat{\varepsilon}) \| \leq C \sqrt{(n^{-1} \hat{\varepsilon}' \hat{Z}) \Omega^{-1} (n^{-1} \hat{Z}' \hat{\varepsilon})} = O_p(\sqrt{r_n/n})
\]

Because the eigenvalues of $\Omega$ are bounded below and above w.p.a. 1, it is also true that $\| n^{-1} \hat{Z}' \hat{\varepsilon} \| = O_p(\sqrt{r_n/n})$. Similarly, $\| n^{-1} \hat{W}' \hat{\varepsilon} \| = O_p(\sqrt{m_n/n})$. Using this result and the inequality $\|A B \|^2 \leq \| A \|^2 \| B \|^2$, get
\begin{align*}
&\Big{\|} \hat{\varepsilon}' P_{\hat{W}} \hat{Z} (n \hat{\Omega})^{-1} \hat{Z}' \hat{\varepsilon} \Big{\|} = \Big{\|} \hat{\varepsilon}' \hat{W} (\hat{W}'\hat{W})^{-1} \hat{W}' \hat{Z} (n \hat{\Omega})^{-1} \hat{T}' \hat{\varepsilon} \Big{\|} \\
&= \Big{\|} (nT^2) \left( \hat{\varepsilon}' \hat{W}/(nT) \right) \left( \hat{W}'\hat{W}/(nT) \right)^{-1} \left( \hat{W}' \hat{T}/(nT) \right) \hat{\Omega}^{-1} \left( \hat{T}' \hat{\varepsilon}/(nT) \right) \Big{\|} \\
&\leq C n \Big{\|} \left( \hat{\varepsilon}' \hat{W}/(nT) \right) \left( \hat{W}' \hat{Z}/(nT) \right) \left( \hat{Z}' \hat{\varepsilon}/(nT) \right) \Big{\|} \\
&\leq C n \Big{\|} \hat{\varepsilon}' \hat{W}/(nT) \Big{\|} \; \Big{\|} \hat{W}' \hat{Z}/(nT) \Big{\|} \; \Big{\|} \hat{Z}' \hat{\varepsilon}/(nT) \Big{\|} \\
&= n O_p(\sqrt{m_n/n}) O_p(\zeta(k_n) \sqrt{k_n/n}) O_p(\sqrt{r_n/n}) = O_p(\zeta(k_n) \sqrt{m_n k_n r_n/n})
\end{align*}

In turn,
\begin{align*}
&\Big{\|} \hat{\varepsilon}' P_{\hat{W}} \hat{Z} (n \hat{\Omega})^{-1} \hat{Z}' P_{\hat{W}} \hat{\varepsilon} \Big{\|} = \Big{\|} \hat{\varepsilon} \hat{W} (\hat{W}'\hat{W})^{-1} \hat{W}' \hat{Z} (n \hat{\Omega})^{-1} \hat{Z}' \hat{W} (\hat{W}'\hat{W})^{-1} \hat{W}' \hat{\varepsilon} \Big{\|} \\
&= \Big{\|} (nT^2) \left( \hat{\varepsilon}' \hat{W}/(nT) \right) \left( \hat{W}' \hat{W}/(nT) \right)^{-1} \left( \hat{W}' \hat{Z}/(nT) \right) \hat{\Omega}^{-1} \\
&\left( \hat{Z}' \hat{W}/(nT) \right) \left( \hat{W}' \hat{W}/(nT) \right)^{-1} \left( \hat{W}' \hat{\varepsilon}/(nT) \right) \Big{\|} \\
&\leq C n \Big{\|} \left( \hat{\varepsilon}' \hat{W}/(nT) \right) \left( \hat{W}' \hat{Z}/(nT) \right) \left( \hat{Z}' \hat{W}/(nT) \right) \left( \hat{W}' \hat{\varepsilon}/(nT) \right) \Big{\|} \\
&\leq C n \Big{\|} \hat{\varepsilon}' \hat{W}/(nT) \Big{\|} \; \Big{\|} \hat{W}' \hat{Z}/(nT) \Big{\|} \; \Big{\|} \hat{Z}' \hat{W}/(nT) \Big{\|} \; \Big{\|} \hat{W}' \hat{\varepsilon}/(nT) \Big{\|} \\
&= n O_p(\sqrt{m_n/n}) O_p(\zeta(k_n) \sqrt{k_n/n}) O_p(\zeta(k_n) \sqrt{k_n/n}) O_p(\sqrt{m_n/n}) = O_p(\zeta(k_n)^2 m_n k_n/n)
\end{align*}

Thus,
\begin{align}\label{eqn_remainders}
\begin{split}
&(\hat{\varepsilon} + \hat{R})' \tilde{Z}  \left( n \hat{\Omega} \right)^{-1} \tilde{Z}' (\hat{\varepsilon} + \hat{R}) = \hat{\varepsilon}' \hat{Z} (n \hat{\Omega})^{-1} \hat{Z}' \hat{\varepsilon} + O_p(n m_n^{-2\alpha}) + O_p(n^{1/2} m_n^{-\alpha}) \\
&+ O_p(\zeta(k_n) \sqrt{m_n k_n r_n/n}) + O_p(\zeta(k_n)^2 m_n k_n/n) = \hat{\varepsilon}' \hat{Z} (n \hat{\Omega})^{-1} \hat{Z}' \hat{\varepsilon} + o_p(\sqrt{r_n})
\end{split}
\end{align}

Step 3. Deal with the leading term. 

As shown in Lemma~\ref{diff_r_n_small},
\begin{align}\label{eqn_t_r_3_hc}
\frac{n (n^{-1} \hat{\varepsilon}' \hat{Z}) \hat{\Omega}^{-1} (n^{-1} \hat{Z}' \hat{\varepsilon}) - n (n^{-1} \hat{\varepsilon}' \hat{Z}) \Omega^{-1} (n^{-1} \hat{Z}' \hat{\varepsilon})}{\sqrt{r_n}} \overset{p}{\to} 0
\end{align}

Note that
\[
\hat{\varepsilon}' \hat{Z} (n \Omega)^{-1} \hat{Z}' \hat{\varepsilon} = n^{-1} \sum_{i=1}^{n}{\sum_{j = 1}^{n}{\hat{\varepsilon}_i' \hat{Z}_i \Omega^{-1} \hat{Z}_j' \hat{\varepsilon}_j}} = n^{-1} \sum_{i=1}^{n}{\hat{\varepsilon}_i' \hat{Z}_i \Omega^{-1} \hat{Z}_i' \hat{\varepsilon}_i} + 2 n^{-1} \sum_{i=1}^{n-1}{\sum_{j > i}{\hat{\varepsilon}_i' \hat{Z}_i \Omega^{-1} \hat{Z}_j' \hat{\varepsilon}_j}}
\]

Thus,
\begin{align}\label{eqn_t}
t_{*} = \frac{\xi_{*} - r_n}{\sqrt{2 r_n}} = \frac{n^{-1} \sum_{i=1}^{n}{\hat{\varepsilon}_i' \hat{Z}_i \Omega^{-1} \hat{Z}_i' \hat{\varepsilon}_i} - r_n}{\sqrt{2 r_n}} + \frac{\sqrt{2} \sum_{i=1}^{n-1}{\sum_{j > i}{\hat{\varepsilon}_i' \hat{Z}_i \Omega^{-1} \hat{Z}_j' \hat{\varepsilon}_j}}}{\sqrt{n^2 r_n}} = t_1 + t_2,
\end{align}
where
\begin{align*}
t_1 &= \frac{n^{-1} \sum_{i=1}^{n}{\hat{\varepsilon}_i' \hat{Z}_i \Omega^{-1} \hat{Z}_i' \hat{\varepsilon}_i} - r_n}{\sqrt{2 r_n}} \\
t_2 &= \frac{\sqrt{2} \sum_{i=1}^{n-1}{\sum_{j > i}{\hat{\varepsilon}_i' \hat{Z}_i \Omega^{-1} \hat{Z}_j' \hat{\varepsilon}_j}}}{\sqrt{n^2 r_n}}
\end{align*}

Note that
\[
E[t_1] = \frac{E[\hat{\varepsilon}_i' \hat{Z}_i \Omega^{-1} \hat{Z}_i' \hat{\varepsilon}_i] - r_n}{\sqrt{2 r_n}} = \frac{E[\chi_i' \Omega^{-1} \chi_i] - r_n}{\sqrt{2 r_n}} = 0,
\]
because
\[
E[\chi_i' \Omega^{-1} \chi_i] = r_n
\]

Next,
\begin{align*}
Var(t_1) &\leq E[(\chi_i' \Omega^{-1} \chi_i)^2]/(2n r_n) \leq C E[\| \chi_i \|^4]/(n r_n) \leq C E[\| \hat{\varepsilon}_i \|^4 \| \hat{Z}_i \|^4]/(n r_n) \\
&\leq C E[\| \hat{Z}_i \|^4]/(n r_n) \leq C \zeta(r_n)^2 r_n/(2n r_n) = C \zeta(r_n)^2/n \to 0,
\end{align*}
so, by Markov's inequality, $t_1 \overset{p}{\to} 0$.

Next, note that $t_2 = \sum_{i=1}^{n-1}{\sum_{j > i}{H_n(\chi_i,\chi_j)}}$, where
\[
H_n(\chi_i,\chi_j) = \sqrt{\frac{2}{n^2 r_n}} \chi_i' \Omega^{-1} \chi_j = \sqrt{\frac{2}{n^2 r_n}} \hat{\varepsilon}_i' \hat{Z}_i \Omega^{-1} \hat{Z}_j' \hat{\varepsilon}_j
\]
Then
\begin{align*}
G_n(u,v) &= E[H_n(\chi_1,u) H_n(\chi_1,v)] = \frac{2}{n^2 r_n} E[\chi_1' \Omega^{-1} u \chi_1' \Omega^{-1} v] \\
&= \frac{2}{n^2 r_n} E[u' \Omega^{-1} \chi_1 \chi_1' \Omega^{-1} v] = \frac{2}{n^2 r_n} u' \Omega^{-1} v = \sqrt{\frac{2}{n^2 r_n}} H_n(u,v)
\end{align*}

Note that $E[H_n(\chi_1,\chi_2)|\chi_1] = \chi_1' \Omega E[\chi_2] = 0$ and that
\begin{align*}
E[H_n(\chi_1,\chi_2)^2] &= \frac{2}{n^2 r_n} E[\chi_1' \Omega^{-1} \chi_2 \chi_1' \Omega^{-1} \chi_2] = \frac{2}{n^2 r_n} E[\chi_1' \Omega^{-1} \chi_2 \chi_2' \Omega^{-1} \chi_1] = \frac{2}{n^2 r_n} E[\chi_1' \Omega^{-1} \chi_1] \\
&= \frac{2}{n^2 r_n} E[tr(\chi_1' \Omega^{-1} \chi_1)] = \frac{2}{n^2 r_n} E[tr(\chi_1 \chi_1' \Omega^{-1})]  = \frac{2}{n^2}
\end{align*}

Thus, we have:
\[
\frac{E[G_n(\chi_1,\chi_2)^2]}{\{E[H_n(\chi_1,\chi_2)^2]\}^2} = \frac{(2/n^2 r_n) (2/n^2)}{(4/n^2)} = \frac{1}{r_n} \to 0
\]
and, using the Cauchy-Schwartz inequality,
\begin{align*}
\frac{n^{-1} E[H_n(\chi_1,\chi_2)^4]}{\{E[H_n(\chi_1,\chi_2)^2]\}^2} &= \frac{(4/n^5 r_n^2) E[(\chi_1' \Omega^{-1} \chi_2)^4]}{(4/n^4)} \leq \frac{E[(\chi_1' \Omega^{-1} \chi_1)^2 (\chi_2' \Omega^{-1} \chi_2)^2]}{n r_n^2} \\
&= \frac{E[(\chi_1' \Omega^{-1} \chi_1)^2]^2}{n r_n^2} = \left[ \frac{E[(\chi_1' \Omega^{-1} \chi_1)^2]}{r_n \sqrt{n}}  \right]^2 \to 0
\end{align*}

Thus, the conditions of Theorem 1 in \citet{hall_1984} hold, and
\begin{align}\label{eqn_t2}
t_2 = \sqrt{\frac{2}{n^2 r_n}} \sum_{i=1}^{n-1}{\sum_{j > i}{\hat{\varepsilon}_i' \hat{Z}_i \Omega^{-1} \hat{Z}_j' \hat{\varepsilon}_j}} \overset{d}{\to} N(0,1)
\end{align}

The result of the theorem now follows from equations~\ref{eqn_remainders}, \ref{eqn_t_r_3_hc}, \ref{eqn_t}, and~\ref{eqn_t2}. \qed

\subsection{Proof of Theorem~\ref{asy_distr_t_r_n_hc_boot}}

The wild bootstrap test statistic is given by
\[
\xi_{HC}^* = \left( \sum_{i=1}^{n}{\tilde{e}_i^{* \prime} \tilde{Z}_i} \right) \left( \sum_{i=1}^{n}{\tilde{Z}_i' \tilde{e}_i^* \tilde{e}_i^{* \prime} \tilde{Z}_i} \right)^{-1} \left( \sum_{i=1}^{n}{\tilde{Z}_i' \tilde{e}_i^*} \right) = \tilde{e}^{* \prime} \tilde{Z}  \left( n \hat{\Omega}^* \right)^{-1} \tilde{Z}' \tilde{e}^*,
\]
where $\hat{\Omega}^* = n^{-1} \sum_{i=1}^{n}{\tilde{Z}_i' \tilde{e}_i^* \tilde{e}_i^{* \prime} \tilde{Z}_i}$.

Note that the bootstrap data is generated as
\[
\hat{Y}^* = \hat{W} \tilde{\beta}_1 + \hat{\varepsilon}^*
\]

The bootstrap residuals are given by $\tilde{e}^* = M_{\hat{W}} \hat{Y}^* = M_{\hat{W}} \hat{\varepsilon}^*$. Thus, the bootstrap test statistic can be rewritten as
\[
\xi_{HC}^* = \hat{\varepsilon}^{* \prime} M_{\hat{W}} \tilde{Z}  \left( n \hat{\Omega}^* \right)^{-1} \tilde{Z}' M_{\hat{W}} \hat{\varepsilon}^*
\]

The rest of the proof is very similar to the proof of Theorem~\ref{asy_distr_t_r_n_hc}, so I only provide a sketch here. First, one can show that
\[
\frac{\hat{\varepsilon}^{* \prime} M_{\hat{W}} \tilde{Z}  \left( n \hat{\Omega}^* \right)^{-1} \tilde{Z}' M_{\hat{W}} \hat{\varepsilon}^* - \hat{\varepsilon}^{* \prime} Z  \left( n \Omega^* \right)^{-1} Z \hat{\varepsilon}^*}{\sqrt{r_n}} \overset{p}{\to} 0,
\]
where $\Omega^* = n^{-1} \sum_{i=1}^{n}{\hat{Z}_i' \tilde{e}_i \tilde{e}_i' \hat{Z}_i}$. Next, one can deal with the leading term $\hat{\varepsilon}^{* \prime} Z  \left( n \Omega^* \right)^{-1} Z \hat{\varepsilon}^*$ as in Step 3 of the proof of Theorem~\ref{asy_distr_t_r_n_hc}, but now conditional on the data $\mathcal{Z}_{n,T}$. The result of the theorem follows.

\subsection{Proof of Theorem~\ref{global_alternative_t_r}}

Recall that $\hat{\Omega} = n^{-1} \sum_{i=1}^{n}{\tilde{Z}_i' \tilde{\Sigma}_T \tilde{Z}_i}$ in the homoskedastic case and $\hat{\Omega} = n^{-1} \sum_{i=1}^{n}{\tilde{Z}_i' \tilde{e}_i \tilde{e}_i' \tilde{Z}_i}$ in the heteroskedastic case. Next, note that
\[
\sum_{i=1}^{n}{\tilde{Z}_i' \tilde{e}_i} = \tilde{Z}' \tilde{e}
\]
where $\tilde{Z} = (\tilde{Z}_1',...,\tilde{Z}_n')'$ is $nT \times r_n$, $\tilde{e} = (\tilde{e}_1',...,\tilde{e}_n')'$ is $nT \times 1$.

Then under homoskedasticity
\[
\xi = \left( \sum_{i=1}^{n}{\tilde{e}_i' \tilde{Z}_i} \right) \left( \sum_{i=1}^{n}{\tilde{Z}_i' \tilde{\Sigma}_T \tilde{Z}_i} \right)^{-1} \left( \sum_{i=1}^{n}{\tilde{Z}_i' \tilde{e}_i} \right) = n \left( n^{-1} \tilde{e}' \tilde{Z} \right) \hat{\Omega}^{-1} \left( n^{-1} \tilde{Z}' \tilde{e} \right),
\]
while under heteroskedasticity
\[
\xi_{HC} = \left( \sum_{i=1}^{n}{\tilde{e}_i' \tilde{Z}_i} \right) \left( \sum_{i=1}^{n}{\tilde{Z}_i' \tilde{e}_i \tilde{e}_i' \tilde{Z}_i} \right)^{-1} \left( \sum_{i=1}^{n}{\tilde{Z}_i' \tilde{e}_i} \right) = n \left( n^{-1} \tilde{e}' \tilde{Z} \right) \hat{\Omega}^{-1} \left(  n^{-1} \tilde{Z}' \tilde{e} \right) 
\]

Also note that
\[
\frac{\sqrt{r_n}}{n} \frac{n \left( n^{-1} \tilde{e}' \tilde{Z} \right) \hat{\Omega}^{-1} \left(  n^{-1} \tilde{Z}' \tilde{e} \right)   - r_n}{\sqrt{2 r_n}} = \frac{1}{\sqrt{2}} \left( n^{-1} \tilde{e}' \tilde{Z} \right) \hat{\Omega}^{-1} \left(  n^{-1} \tilde{Z}' \tilde{e} \right)  + T_2,
\]
where $T_2 = - r_n/(n \sqrt{2}) \to 0$.

Hence, it suffices to show that $\left( n^{-1} \tilde{e}' \tilde{Z} \right) \hat{\Omega}^{-1} \left(  n^{-1} \tilde{Z}' \tilde{e} \right)   \overset{p}{\to} \Delta$.

Next, note that due to the projection nature of the series estimators, $\tilde{e} = M_{\hat{W}} \hat{Y} = M_{\hat{W}} \hat{\varepsilon}^* + M_{\hat{W}} \hat{R}^*$. Hence,
\begin{align*}
&\left( n^{-1} \tilde{e}' \tilde{Z} \right) \hat{\Omega}^{-1} \left(  n^{-1} \tilde{Z}' \tilde{e} \right)  = \left( n^{-1} (M_{\hat{W}} \hat{\varepsilon}^* + M_{\hat{W}} \hat{R}^*) ' \tilde{Z} \right) \hat{\Omega}^{-1} \left(  n^{-1} \tilde{Z}' (M_{\hat{W}} \hat{\varepsilon}^* + M_{\hat{W}} \hat{R}^*) \right) \\
&= \left( n^{-1} (\hat{\varepsilon}^* + \hat{R}^*)' M_{\hat{W}} \tilde{Z} \right) \hat{\Omega}^{-1} \left(  n^{-1} \tilde{Z}' M_{\hat{W}} (\hat{\varepsilon}^* + \hat{R}^*) \right) \\
&= \left( n^{-1} (\hat{\varepsilon}^* + \hat{R}^*)' M_{\hat{W}} \hat{Z} \right) \hat{\Omega}^{-1} \left(  n^{-1} \hat{Z}' M_{\hat{W}} (\hat{\varepsilon}^* + \hat{R}^*) \right) 
\end{align*}

Thus, it suffices to show that
\[
n \left( n^{-1} (\hat{\varepsilon}^* + \hat{R}^*)' M_{\hat{W}} \hat{Z} \right) \hat{\Omega}^{-1} \left(  n^{-1} \hat{Z}' M_{\hat{W}} (\hat{\varepsilon}^* + \hat{R}^*) \right)  \overset{p}{\to} \Delta
\]

Next,
\begin{align*}
&n \left( n^{-1} (\hat{\varepsilon}^* + \hat{R}^*)' M_{\hat{W}} \hat{Z} \right) \hat{\Omega}^{-1} \left(  n^{-1} \hat{Z}' M_{\hat{W}} (\hat{\varepsilon}^* + \hat{R}^*) \right) = \left( n^{-1} \hat{\varepsilon}^{* \prime} M_{\hat{W}} \hat{Z} \right) \hat{\Omega}^{-1} \left( n^{-1} \hat{Z}' M_{\hat{W}} \hat{\varepsilon}^* \right) \\
&+ 2 \left( n^{-1} \hat{R}^{* \prime} M_{\hat{W}} \hat{Z} \right) \hat{\Omega}^{-1} \left( n^{-1} \hat{Z}' M_{\hat{W}} \hat{\varepsilon}^* \right) + \left( n^{-1} \hat{R}^{* \prime} M_{\hat{W}} \hat{Z} \right) \hat{\Omega}^{-1} \left( n^{-1} \hat{Z}' M_{\hat{W}} \hat{R}^* \right)
\end{align*}

Similarly to the proof of Theorem~\ref{asy_distr_t_r_n_hc}, but using the fact that $\sup_{x \in \mathcal{X}}{R^*(x)} = o(1)$ instead of $\sup_{x \in \mathcal{X}}{R(x)} = O(m_n^{-\alpha})$,
\begin{align*}
\left( n^{-1} \hat{R}^{* \prime} M_{\hat{W}} \hat{Z} \right) \hat{\Omega}^{-1} \left( n^{-1} \hat{Z}' M_{\hat{W}} \hat{\varepsilon}^* \right)  \leq C \hat{R}^{* \prime} \hat{\varepsilon}^*/(n \tilde{\sigma}^2) = O_p(n^{-1/2}) o_p(1) = o_p(1)
\end{align*}
and
\[
\left( n^{-1} \hat{R}^{* \prime} M_{\hat{W}} \hat{Z} \right) \hat{\Omega}^{-1} \left( n^{-1} \hat{Z}' M_{\hat{W}} \hat{R}^* \right) \leq C \hat{R}^{* \prime} \hat{R}^*/(n \tilde{\sigma}^2) = o_p(1)
\]

Thus,
\[
n \left( n^{-1} (\hat{\varepsilon}^* + \hat{R}^*)' M_{\hat{W}} \hat{Z} \right) \hat{\Omega}^{-1} \left(  n^{-1} \hat{Z}' M_{\hat{W}} (\hat{\varepsilon}^* + \hat{R}^*) \right) = \left( n^{-1} \hat{\varepsilon}^{* \prime} M_{\hat{W}} \hat{Z} \right) \hat{\Omega}^{-1} \left( n^{-1} \hat{Z}' M_{\hat{W}} \hat{\varepsilon}^* \right)  + o_p(1)
\]

Next, given that, as shown in the proof of Theorem~\ref{asy_distr_t_r_n_hc}, $M_{\hat{W}} \hat{Z}  = \hat{Z} + o_p(1)$ and the eigenvalues of $\hat{\Omega}$ are bounded above and below w.p.a. 1,
\[
\left( n^{-1} \hat{\varepsilon}^{* \prime} M_{\hat{W}} \hat{Z} \right) \hat{\Omega}^{-1} \left( n^{-1} \hat{Z}' M_{\hat{W}} \hat{\varepsilon}^* \right) = \left( n^{-1} \hat{\varepsilon}^{* \prime} \hat{Z} \right) \hat{\Omega}^{-1} \left( n^{-1} \hat{Z}'  \hat{\varepsilon}^* \right) + o_p(1)
\]

Next,
\begin{align*}
&\Big{|} (n^{-1} \hat{\varepsilon}^{* \prime} \hat{Z}) (\hat{\Omega}^{-1} - \Omega^{*-1}) (n^{-1} \hat{Z}' \hat{\varepsilon}^{*}) \Big{|} \leq \Big{|} (n^{-1} \hat{\varepsilon}^{* \prime} \hat{Z}) (\Omega^{*-1} (\hat{\Omega} - \Omega^*) \hat{\Omega}^{*-1} (\hat{\Omega} - \Omega^*) \Omega^{*-1}) (n^{-1} \hat{Z}' \hat{\varepsilon}^{*}) \Big{|} \\
&+\Big{|} (n^{-1} \hat{\varepsilon}^{* \prime} \hat{Z}) (\Omega^{*-1} (\hat{\Omega} - \Omega^*) \Omega^{*-1}) (n^{-1} \hat{Z}' \hat{\varepsilon}^{*}) \Big{|} \leq \| \Omega^{*-1} n^{-1} \hat{Z}' \hat{\varepsilon}^{*} \|^2 (\| \hat{\Omega} - \Omega^* \| + C \| \hat{\Omega} - \Omega^* \|^2) = o_p(1)
\end{align*}

Thus, $(n^{-1} \hat{\varepsilon}^{* \prime} \hat{Z}) \hat{\Omega}^{-1} (n^{-1} \hat{Z}' \hat{\varepsilon}^*) = (n^{-1} \hat{\varepsilon}^{* \prime} \hat{Z}) \Omega^{*-1} (n^{-1} \hat{Z}' \hat{\varepsilon}^{*}) + o_p(1)$.

To complete the proof, note that $Var(\hat{Z}_i' \hat{\varepsilon}_i^*) \leq \Omega^*$, because $\Omega^* = E[\hat{Z}_i' \hat{\varepsilon}_i \hat{\varepsilon}_i' \hat{Z}_i]$. Then
\begin{align*}
&E[(n^{-1} \hat{Z}' \hat{\varepsilon}^* - E[\hat{Z}_i' \hat{\varepsilon}_i^*])' \Omega^{*-1} (n^{-1} \hat{Z}' \hat{\varepsilon}^* - E[\hat{Z}_i' \hat{\varepsilon}_i^*])] \\
&\leq E[(n^{-1} \hat{Z}' \hat{\varepsilon}^* - E[\hat{Z}_i' \hat{\varepsilon}_i^*])' Var(\hat{Z}_i' \hat{\varepsilon}_i^*)^{-1} (n^{-1} \hat{Z}' \hat{\varepsilon}^* - E[\hat{Z}_i' \hat{\varepsilon}_i^*])] \\
&= E[tr\left(Var(\hat{Z}_i' \hat{\varepsilon}_i^*)^{-1} (n^{-1} \hat{Z}' \hat{\varepsilon}^* - E[\hat{Z}_i' \hat{\varepsilon}_i^*]) (n^{-1} \hat{Z}' \hat{\varepsilon}^* - E[\hat{Z}_i' \hat{\varepsilon}_i^*])'\right)] =tr(I_{r_n})/n = r_n/n \to 0
\end{align*}

Thus,
\begin{align*}
&\Big{|}  (n^{-1} \hat{\varepsilon}^{* \prime} \hat{Z}) \Omega^{*-1} (n^{-1} \hat{Z}' \hat{\varepsilon}^{*}) - E[\hat{\varepsilon}_i^{*\prime} \hat{Z}_i] \Omega^{*-1} E[\hat{Z}_i' \hat{\varepsilon}_i^*] \Big{|}  \\
&\leq \Big{|}  (n^{-1} \hat{Z}' \hat{\varepsilon}^* - E[\hat{Z}_i' \hat{\varepsilon}_i^*])' \Omega^{*-1} (n^{-1} \hat{Z}' \hat{\varepsilon}^* - E[\hat{Z}_i' \hat{\varepsilon}_i^*]) \Big{|}  + 2 \Big{|}  E[\hat{\varepsilon}_i^{* \prime} \hat{Z}_i] \Omega^{*-1} (n^{-1} \hat{Z}' \hat{\varepsilon}^* - E[\hat{Z}_i' \hat{\varepsilon}_i^*]) \Big{|} \\
&\leq o_p(1) + 2 \sqrt{E[\hat{\varepsilon}_i^* \hat{Z}_i'] \Omega^{*-1} E[\hat{Z}_i \hat{\varepsilon}_i^*]} \sqrt{(n^{-1} \hat{Z}' \hat{\varepsilon}^* - E[\hat{Z}_i' \hat{\varepsilon}_i^*])' \Omega^{*-1} (n^{-1} \hat{Z}' \hat{\varepsilon}^* - E[\hat{Z}_i' \hat{\varepsilon}_i^*])} \\
&= o_p(1) + 2 \sqrt{\Delta} o_p(1) = o_p(1)
\end{align*}

Combining the results above, $\left( n^{-1} (\hat{\varepsilon}^* + \hat{R}^*)' M_{\hat{W}} \hat{Z} \right) \hat{\Omega}^{-1} \left(  n^{-1} \hat{Z}' M_{\hat{W}} (\hat{\varepsilon}^* + \hat{R}^*) \right)  \overset{p}{\to} \Delta$. \qed

\subsection{Auxiliary Lemmas}

\begin{lemma}[\citet{donald_et_al_2003}, Lemma A.2]\label{series_normalization}

If Assumption \ref{series_norms_eigenvalues} is satisfied then it can be assumed without loss of generality that $\hat{P}^{k}(x) = \bar{P}^{k}(x)$ and that $E[\bar{P}^{k}(X_{it}) \bar{P}^{k}(X_{it})'] = I_{k}$.

\end{lemma}

\begin{remark}
This normalization is common in the literature on series estimation, when all elements of $P^{k}(X_i)$ are used to estimate the model. In my setting, $P^{k}(X_i)$ is partitioned into $W^{m}(X_i)$, used in estimation, and $T^{r}(X_i)$, used in testing. The normalization implies that $W_i$ and $T_i$ are orthogonal to each other. This can be justified as follows. Suppose that $W_i$ and $T_i$ are not orthogonal. Then one can take all elements of $(W_i', T_i')'$ and apply the Gram-Schmidt process to them. Because the orthogonalization process is sequential, it will yield the vector $(W_i^{0},T_i^{0})'$ such that $W_i^{0}$ spans the same space as $W_i$, $T_i^{0}$ spans the same space as $T_i$, and $W_i^{0}$ and $T_i^{0}$ are orthogonal. Thus, the normalization is indeed without loss of generality.
\end{remark}

\begin{lemma}[\citet{baltagi_li_2002}, Theorem 2.2]\label{series_f_rates}

Let $f(x) = f(x,\theta_0,h_0)$, $f_i = f(X_i)$, $\tilde{f}(x)= f(x, \tilde{\theta}, \tilde{h}) = W^{m_n}(x)' \tilde{\beta}_0$, and $\tilde{f}_i = \tilde{f}(X_i) = W_i' \tilde{\beta}_0$. Under Assumptions \ref{dgp}, \ref{series_norms_eigenvalues}, and \ref{series_approx}, the following is true:
\[
\frac{1}{n}\sum_{i=1}^{n}{(\tilde{f}_i - f_i)^2} = O_p(m_n/n + m_n^{-2\alpha})
\]

\end{lemma}

\begin{lemma}\label{lemma_din_1}

Suppose that Assumption~\ref{assumption_din_1} is satisfied and $E[\hat{\varepsilon}_{i} \hat{\varepsilon}_{i}']$ is finite. If $E[\hat{\varepsilon}_{i} | \hat{X}_i] = 0$ then $E[\hat{P}_{i}' \hat{\varepsilon}_{i}] = 0$ for all $k$. Furthermore, if $E[\hat{\varepsilon}_{i} | X_i] \neq 0$ then $E[\hat{P}_{i}' \hat{\varepsilon}_{i}] \neq 0$ for all $k$ large enough.
\end{lemma}

\begin{proof}

The proof is similar to the proof of Lemma 2.1 in \citet{donald_et_al_2003} and is thus omitted.

\end{proof}

\begin{lemma}\label{lemma_PP}
If Assumptions of Theorem~\ref{asy_distr_t_r_n_hc} hold, then $\| \hat{P}'\hat{P}/(nT)- I_{k_n}\| = O_p(\zeta(k_n) \sqrt{k_n/n})$, $\| \hat{W}'\hat{W}/(nT) - I_{m_n}\| = O_p(\zeta(m_n) \sqrt{m_n/n})$ and $\| \hat{Z}'\hat{Z}/(nT) - I_{r_n}\| = O_p(\zeta(r_n) \sqrt{r_n/n})$. Moreover, $\| \hat{W}'\hat{Z}/(nT) \| = O_p(\zeta(k_n) \sqrt{k_n/n})$.

\end{lemma}

\begin{proof}

By Lemma A.1 in \citet{baltagi_li_2002}, $\| \hat{P}'\hat{P}/(nT)- I_{k_n}\| = O_p(\zeta(k_n) \sqrt{k_n/n})$. Similarly, it can be shown that $\| \hat{W}'\hat{W}/(nT) - I_{m_n}\| = O_p(\zeta(m_n) \sqrt{m_n/n})$ and $\| \hat{Z}'\hat{Z}/(nT) - I_{r_n}\| = O_p(\zeta(r_n) \sqrt{r_n/n})$. Moreover, note that
\[
\hat{P}'\hat{P}/(nT) - I_{k_n} = \begin{pmatrix}
\hat{W}'\hat{W}/(nT) & \hat{W}'\hat{Z}/(nT) \\
\hat{Z}'\hat{W}/(nT) & \hat{Z}'\hat{Z}/(nT)
\end{pmatrix} - 
\begin{pmatrix}
I_{m_n} & \textbf{0}_{m_n \times r_n} \\
\textbf{0}_{r_n \times m_n} & I_{r_n}
\end{pmatrix}
\]

Hence, $\| \hat{W}'\hat{Z}/(nT) \| = O_p(\zeta(k_n) \sqrt{k_n/n})$.

\end{proof}

\begin{lemma}\label{omegas_T_hc}
In the homoskedastic case, let $\tilde{\Omega} = n^{-1} \sum_{i=1}^{n}{\hat{Z}_i' \tilde{\Sigma}_T \hat{Z}_i}$, $\breve{\Omega} = n^{-1} \sum_{i=1}^{n}{\hat{Z}_i' \Sigma_T \hat{Z}_i}$.

In the heteroskedastic case, let $\tilde{\Omega} = n^{-1} \sum_{i=1}^{n}{\hat{Z}_i' \tilde{e}_i \tilde{e}_i' \hat{Z}_i}$, $\breve{\Omega} = n^{-1} \sum_{i=1}^{n}{\hat{Z}_i' \hat{\varepsilon}_i \hat{\varepsilon}_i' \hat{Z}_i}$, and $\bar{\Omega} = n^{-1} \sum_{i=1}^{n}{\hat{Z}_i' \Sigma_i \hat{Z}_i}$, where $\Sigma_i = E[\hat{\varepsilon}_i \hat{\varepsilon}_i' | X_i]$.

Let $\Omega = E[\hat{Z}_i \Sigma_i \hat{Z}_i']$.

 Suppose that Assumptions \ref{errors_unified}(ii), \ref{series_norms_eigenvalues}, and \ref{series_approx} are satisfied. Then
\begin{align*}
\| \tilde{\Omega} - \breve{\Omega} \| &=  O_p \left( \zeta(r_n)^2 (m_n/n + m_n^{-2\alpha}) \right) \\
\| \breve{\Omega} - \bar{\Omega} \| &= O_p(\zeta(r_n) r_n^{1/2}/n^{1/2}) \\
\| \bar{\Omega} - \Omega \| &= O_p(\zeta(r_n) r_n^{1/2}/n^{1/2})
\end{align*}

If Assumption~\ref{errors_unified}(i) is also satisfied then $1/C \leq \lambda_{\min}(\Omega) \leq \lambda_{\max}(\Omega) \leq C$, and if $\zeta(r_n)^2 (m_n/n + m_n^{-2\alpha}) \to 0$ and $\zeta(r_n) r_n^{1/2}/n^{1/2} \to 0$, then w.p.a. 1, $1/C \leq \lambda_{\min}(\tilde{\Omega}) \leq \lambda_{\max}(\tilde{\Omega}) \leq C$ and $1/C \leq \lambda_{\min}(\bar{\Omega}) \leq \lambda_{\max}(\bar{\Omega}) \leq C$.

Moreover, if $\| \hat{\Omega} - \tilde{\Omega} \| = o_p(1)$, then $1/C \leq \lambda_{\min}(\hat{\Omega}) \leq \lambda_{\max}(\hat{\Omega}) \leq C$.

\end{lemma}

\begin{proof}[Sketch of the Proof of Lemma~\ref{omegas_T_hc}]

Note that in the homoskedastic case, all matrices involved in the lemma can be written as
\begin{align*}
\hat{\Omega} &= n^{-1} \sum_{i=1}^{n}{\tilde{Z}_i' \tilde{\Sigma}_T \tilde{Z}_i} = \sum_{t=1}^{T}{\sum_{s=1}^{T}{\tilde{\sigma}_{ts} n^{-1} \sum_{i=1}^{n}{ \tilde{Z}_{it} \tilde{Z}_{is}'}}} \\
\tilde{\Omega} &= n^{-1} \sum_{i=1}^{n}{\hat{Z}_i' \tilde{\Sigma}_T \hat{Z}_i} = \sum_{t=1}^{T}{\sum_{s=1}^{T}{\tilde{\sigma}_{ts} n^{-1} \sum_{i=1}^{n}{ \hat{Z}_{it} \hat{Z}_{is}'}}} \\
\breve{\Omega} &= n^{-1} \sum_{i=1}^{n}{\hat{Z}_i' \Sigma_T \hat{Z}_i} =  \sum_{t=1}^{T}{\sum_{s=1}^{T}{\sigma_{ts} n^{-1} \sum_{i=1}^{n}{ \hat{Z}_{it} \hat{Z}_{is}'}}} 
\end{align*}

In the heteroskedastic case, all matrices involved in the lemma can be written as
\begin{align*}
\hat{\Omega} &= n^{-1} \sum_{i=1}^{n}{\tilde{Z}_i' \tilde{e}_i \tilde{e}_i' \tilde{Z}_i} = \sum_{t=1}^{T}{\sum_{s=1}^{T}{n^{-1} \sum_{i=1}^{n}{\tilde{e}_{it} \tilde{e}_{is} \tilde{Z}_{it} \tilde{Z}_{is}'}}} \\
\tilde{\Omega} &= n^{-1} \sum_{i=1}^{n}{\hat{Z}_i' \tilde{e}_i \tilde{e}_i' \hat{Z}_i} = \sum_{t=1}^{T}{\sum_{s=1}^{T}{n^{-1} \sum_{i=1}^{n}{\tilde{e}_{it} \tilde{e}_{is} \hat{Z}_{it} \hat{Z}_{is}'}}} \\
\breve{\Omega} &= n^{-1} \sum_{i=1}^{n}{\hat{Z}_i' \hat{\varepsilon}_i \hat{\varepsilon}_i' \hat{Z}_i} = \sum_{t=1}^{T}{\sum_{s=1}^{T}{n^{-1} \sum_{i=1}^{n}{\hat{\varepsilon}_{it} \hat{\varepsilon}_{is} \hat{Z}_{it} \hat{Z}_{is}'}}} \\
\bar{\Omega} &= n^{-1} \sum_{i=1}^{n}{\hat{Z}_i' \Sigma_i \hat{Z}_i} = \sum_{t=1}^{T}{\sum_{s=1}^{T}{n^{-1} \sum_{i=1}^{n}{\sigma_{ts} \hat{Z}_{it} \hat{Z}_{is}'}}}
\end{align*}

Because $T$ is finite, Lemma~\ref{omegas_T_hc} can be proved by applying the results from Lemmas A.5 and A.6 in \citet{korolev_2019} element by element.

\end{proof}

\begin{lemma}\label{diff_r_n_small}
If Assumptions of Theorem~\ref{asy_distr_t_r_n_hc} hold, then
\[
\frac{n (n^{-1} \hat{\varepsilon}' \hat{Z}) \hat{\Omega}^{-1} (n^{-1} \hat{Z}' \hat{\varepsilon}) - n (n^{-1} \hat{\varepsilon}' \hat{Z}) \Omega^{-1} (n^{-1} \hat{Z}' \hat{\varepsilon})}{\sqrt{r_n}} \overset{p}{\to} 0
\]
\end{lemma}

\begin{proof}

Note that $\hat{\varepsilon}' \hat{Z} (n \hat{\Omega})^{-1} \hat{Z}' \hat{\varepsilon} = n (n^{-1} \hat{\varepsilon}' \hat{Z}) \hat{\Omega}^{-1} (n^{-1} \hat{Z}' \hat{\varepsilon})$. Then
\begin{align*}
&\Bigg{|} \frac{n (n^{-1} \hat{\varepsilon}' \hat{Z}) \hat{\Omega}^{-1} (n^{-1} \hat{Z}' \hat{\varepsilon})}{\sqrt{2 r_n}} - \frac{n (n^{-1} \hat{\varepsilon}' \hat{Z}) \Omega^{-1} (n^{-1} \hat{Z}' \hat{\varepsilon})}{\sqrt{2r_n}} \Bigg{|} = \Bigg{|} \frac{n (n^{-1} \hat{\varepsilon}' \hat{Z}) (\hat{\Omega}^{-1} - \Omega^{-1}) (n^{-1} \hat{Z}' \hat{\varepsilon})}{\sqrt{2r_n}} \Bigg{|} \\
&\leq \frac{n \| \Omega^{-1} n^{-1} \hat{Z}' \hat{\varepsilon} \|^2 (\| \hat{\Omega} - \Omega \| + C \| \hat{\Omega} - \Omega \|^2)}{\sqrt{2r_n}}
\end{align*}

As shown above, $\| \Omega^{-1} (n^{-1} \hat{Z}' \hat{\varepsilon}) \| = O_p(\sqrt{r_n/n})$.

Then
\begin{align*}
\frac{n \| \Omega^{-1} n^{-1} \hat{Z}' \hat{\varepsilon} \|^2 (\| \tilde{\Omega} - \Omega \| + C \| \tilde{\Omega} - \Omega \|^2)}{\sqrt{2r_n}} = \frac{n O_p(r_n/n) o_p(1/\sqrt{r_n})}{\sqrt{2r_n}} = \frac{o_p(\sqrt{r_n})}{\sqrt{2r_n}} = o_p(1),
\end{align*}
provided that $\|\hat{\Omega} - \Omega \| = o_p(1/\sqrt{r_n})$, which holds under the rate conditions~\ref{rate_cond_r_1_hc}--\ref{rate_cond_r_5_hc} and the high level assumption that $\| \hat{\Omega} - \tilde{\Omega} \| = o_p(r_n^{-1/2})$.

\end{proof}

\clearpage

\bibliography{literature_panel_spec_testing}

\begin{thebibliography}{21}
\newcommand{\enquote}[1]{``#1''}
\expandafter\ifx\csname natexlab\endcsname\relax\def\natexlab#1{#1}\fi

\bibitem[\protect\citeauthoryear{Ai and Li}{Ai and Li}{2008}]{ai_li_2008}
\textsc{Ai, C. and Q.~Li} (2008): \enquote{Semi-parametric and Non-parametric
  Methods in Panel Data Models,} in \emph{The Econometrics of Panel Data:
  Fundamentals and Recent Developments in Theory and Practice}, ed. by
  L.~M{\'a}ty{\'a}s and P.~Sevestre, Berlin, Heidelberg: Springer Berlin
  Heidelberg, 451--478.

\bibitem[\protect\citeauthoryear{An, Hsiao, and Li}{An
  et~al.}{2016}]{an_et_al_2016}
\textsc{An, Y., C.~Hsiao, and D.~Li} (2016): \enquote{Semiparametric Estimation
  of Partially Linear Varying Coefficient Panel Data Models,} in \emph{Essays
  in Honor of Aman Ullah (Advances in Econometrics)}, ed. by
  G.~Gonzalez-Rivera, R.~C. Hill, and T.-H. Lee, Emerald Group Publishing
  Limited, vol.~36, chap.~2, 47--65.

\bibitem[\protect\citeauthoryear{Arellano}{Arellano}{2003}]{arellano_2003}
\textsc{Arellano, M.} (2003): \emph{Panel data econometrics}, Oxford university
  press.

\bibitem[\protect\citeauthoryear{Baltagi}{Baltagi}{2013}]{baltagi_2013}
\textsc{Baltagi, B.} (2013): \emph{Econometric analysis of panel data}, John
  Wiley \& Sons, 5 ed.

\bibitem[\protect\citeauthoryear{Baltagi and Khanti-Akom}{Baltagi and
  Khanti-Akom}{1990}]{baltagi_khanti-akom_1990}
\textsc{Baltagi, B.~H. and S.~Khanti-Akom} (1990): \enquote{On efficient
  estimation with panel data: An empirical comparison of instrumental variables
  estimators,} \emph{Journal of Applied Econometrics}, 5, 401--406.

\bibitem[\protect\citeauthoryear{Baltagi and Li}{Baltagi and
  Li}{2002}]{baltagi_li_2002}
\textsc{Baltagi, B.~H. and D.~Li} (2002): \enquote{Series Estimation of
  Partially Linear Panel Data Models with Fixed Effects,} \emph{Annals of
  Economics and Finance}, 3, 103--116.

\bibitem[\protect\citeauthoryear{Cornwell and Rupert}{Cornwell and
  Rupert}{1988}]{cornwell_rupert_1988}
\textsc{Cornwell, C. and P.~Rupert} (1988): \enquote{Efficient estimation with
  panel data: An empirical comparison of instrumental variables estimators,}
  \emph{Journal of Applied Econometrics}, 3, 149--155.

\bibitem[\protect\citeauthoryear{Davidson and Flachaire}{Davidson and
  Flachaire}{2008}]{davidson_flachaire_2008}
\textsc{Davidson, R. and E.~Flachaire} (2008): \enquote{The wild bootstrap,
  tamed at last,} \emph{Journal of Econometrics}, 146, 162--169.

\bibitem[\protect\citeauthoryear{Donald, Imbens, and Newey}{Donald
  et~al.}{2003}]{donald_et_al_2003}
\textsc{Donald, S.~G., G.~W. Imbens, and W.~K. Newey} (2003):
  \enquote{Empirical likelihood estimation and consistent tests with
  conditional moment restrictions,} \emph{Journal of Econometrics}, 117, 55 --
  93.

\bibitem[\protect\citeauthoryear{Guay and Guerre}{Guay and
  Guerre}{2006}]{guay_guerre_2006}
\textsc{Guay, A. and E.~Guerre} (2006): \enquote{A Data-Driven Nonparametric
  Specification Test for Dynamic Regression Models,} \emph{Econometric Theory},
  22, 543--586.

\bibitem[\protect\citeauthoryear{Hall}{Hall}{1984}]{hall_1984}
\textsc{Hall, P.} (1984): \enquote{Central limit theorem for integrated square
  error of multivariate nonparametric density estimators,} \emph{Journal of
  Multivariate Analysis}, 14, 1 -- 16.

\bibitem[\protect\citeauthoryear{Henderson, Carroll, and Li}{Henderson
  et~al.}{2008}]{henderson_et_al_2008}
\textsc{Henderson, D.~J., R.~J. Carroll, and Q.~Li} (2008):
  \enquote{Nonparametric estimation and testing of fixed effects panel data
  models,} \emph{Journal of Econometrics}, 144, 257 -- 275.

\bibitem[\protect\citeauthoryear{Hong and White}{Hong and
  White}{1995}]{hong_white_1995}
\textsc{Hong, Y. and H.~White} (1995): \enquote{Consistent Specification
  Testing Via Nonparametric Series Regression,} \emph{Econometrica}, 63,
  1133--1159.

\bibitem[\protect\citeauthoryear{Hsiao}{Hsiao}{2003}]{hsiao_2003}
\textsc{Hsiao, C.} (2003): \emph{Analysis of Panel Data}, Econometric Society
  Monographs, Cambridge University Press, 2 ed.

\bibitem[\protect\citeauthoryear{Korolev}{Korolev}{2019}]{korolev_2019}
\textsc{Korolev, I.} (2019): \enquote{{A Consistent LM Type Specification Test
  for Semiparametric Models},} \emph{ArXiv e-prints}.

\bibitem[\protect\citeauthoryear{Li and Wang}{Li and Wang}{1998}]{li_wang_1998}
\textsc{Li, Q. and S.~Wang} (1998): \enquote{A simple consistent bootstrap test
  for a parametric regression function,} \emph{Journal of Econometrics}, 87,
  145--165.

\bibitem[\protect\citeauthoryear{Lin, Li, and Sun}{Lin
  et~al.}{2014}]{lin_et_al_2014}
\textsc{Lin, Z., Q.~Li, and Y.~Sun} (2014): \enquote{A consistent nonparametric
  test of parametric regression functional form in fixed effects panel data
  models,} \emph{Journal of Econometrics}, 178, 167--179.

\bibitem[\protect\citeauthoryear{Mammen}{Mammen}{1993}]{mammen_1993}
\textsc{Mammen, E.} (1993): \enquote{Bootstrap and Wild Bootstrap for High
  Dimensional Linear Models,} \emph{The Annals of Statistics}, 21, 255--285.

\bibitem[\protect\citeauthoryear{Parmeter and Racine}{Parmeter and
  Racine}{2018}]{parmeter_racine_2018}
\textsc{Parmeter, C.~F. and J.~S. Racine} (2018): \enquote{{Nonparametric
  Estimation and Inference for Panel Data Models},} Department of Economics
  Working Papers 2018-02, McMaster University.

\bibitem[\protect\citeauthoryear{Rodriguez-Poo and Soberon}{Rodriguez-Poo and
  Soberon}{2017}]{rodriguez-poo_soberon_2017}
\textsc{Rodriguez-Poo, J.~M. and A.~Soberon} (2017): \enquote{Nonparametric and
  semiparametric panel data models: Recent developments,} \emph{Journal of
  Economic Surveys}, 31, 923--960.

\bibitem[\protect\citeauthoryear{Su and Ullah}{Su and
  Ullah}{2011}]{su_ullah_2011}
\textsc{Su, L. and A.~Ullah} (2011): \enquote{Nonparametric and semiparametric
  panel econometric models: estimation and testing,} in \emph{Handbook of
  empirical economics and finance}, ed. by D.~Giles and A.~Ullah, New York:
  Taylor \& Francis Group, 455--497.

\end{thebibliography}

\end{document}